\documentclass[10pt]{article}
\usepackage[margin=1in]{geometry} % page margin
\usepackage{setspace}
\usepackage{amsfonts,amssymb,amsmath,amsthm}
\usepackage{graphicx,float,afterpage} % for figure presentation
\usepackage{enumerate,color,caption,boxedminipage}
\usepackage{soul} % strike out
\usepackage[ruled,linesnumbered,algoruled,boxed]{algorithm2e}
\usepackage{mwe,tikz}
\usepackage{array,multirow,ctable}
\usepackage[percent]{overpic}	

\usepackage[round]{natbib}
\usepackage{hyperref}
\hypersetup{colorlinks=true,linkcolor=red,citecolor=gray}
\usepackage{url} % not crucial - just used below for the URL 

\newcommand{\blind}{1}

\usepackage{chngcntr,apptools,titlesec}
\titleformat{\section}{\large\normalfont\centering\bfseries}{\thesection.}{1em}{\uppercase}
\titleformat{\subsection}{\normalfont\normalsize\bfseries}{\thesubsection}{1em}{}
\renewcommand{\thesection}{\arabic{section}}
\renewcommand{\thesubsection}{\arabic{section}.\arabic{subsection}}

\newtheorem{lemma}{Lemma}
\newtheorem{theorem}{Theorem}
\newtheorem{proposition}{Proposition}
\newtheorem{corollary}{Corollary}
\theoremstyle{definition}
\newtheorem{definition}{Definition}
\theoremstyle{remark}
\newtheorem{remark}{Remark}

\DeclareMathOperator*{\argmin}{arg~min~}

\newcommand{\vertiii}[1]{{\left\vert\kern-0.25ex\left\vert\kern-0.25ex\left\vert #1 
		\right\vert\kern-0.25ex\right\vert\kern-0.25ex\right\vert}}
\newcommand{\vertii}[1]{{\left\vert\kern-0.25ex\left\vert #1 
		\right\vert\kern-0.25ex\right\vert}}
\newcommand{\smalliii}[1]{{\vert\kern-0.25ex\vert\kern-0.25ex\vert #1 \vert\kern-0.25ex\vert\kern-0.25ex\vert}}
\newcommand{\smallii}[1]{{\vert\kern-0.25ex\vert #1\vert\kern-0.25ex\vert}}

\newcommand{\VAR}{\mathrm{VAR}}
\newcommand{\R}{\mathbb{R}}
\newcommand{\I}{\mathrm{I}}
\newcommand{\E}{\mathbb{E}}
\newcommand{\Cov}{\text{Cov}}
\newcommand{\tr}{\text{trace}}

\newcommand{\rank}{\text{rank}}

\newcommand{\op}{\text{op}}
\newcommand{\F}{\text{F}}
\newcommand{\PP}{\mathbb{P}}

\newcommand{\RERR}{\mathrm{RErr}}
\newcommand{\bunderline}[1]{\underline{#1\mkern-2mu}\mkern2mu }

\makeatletter
\newsavebox{\@brx}
\newcommand{\llangle}[1][]{\savebox{\@brx}{\(\m@th{#1\langle}\)}%
	\mathopen{\copy\@brx\mkern2mu\kern-0.9\wd\@brx\usebox{\@brx}}}
\newcommand{\rrangle}[1][]{\savebox{\@brx}{\(\m@th{#1\rangle}\)}%
	\mathclose{\copy\@brx\mkern2mu\kern-0.9\wd\@brx\usebox{\@brx}}}
\makeatother
\begin{document}

\def\spacingset#1{\renewcommand{\baselinestretch}%
{#1}\small\normalsize} \spacingset{1}
%%%%%%%%%%%%%%%%%%%%%%%%%%%%%%%%%%%%%%%%%%%%%%%%%%%%%%%%%%%%%%%%%%%%%%%%%%%%%%

\if1\blind
{
	\title{\bf\large Approximate Factor Models with Strongly Correlated Idiosyncratic Errors}
	\author{Jiahe Lin \thanks{Jiahe Lin, Department of Statistics, University of Michigan, Ann Arbor, MI 48109 (e-mail: jiahelin@umich.edu)}\\
		Department of Statistics, University of Michigan \vspace*{2mm}\\ 
		and  \vspace*{2mm} \\
		George Michailidis \thanks{George Michailidis, Department of Statistics and the Informatics Institute, University of Florida, Gainesville, FL 32611 (e-mail: gmichail@ufl.edu)} \\
		Department of Statistics and the Informatics Institute, University of Florida}
	\date{\vspace{-5ex}}
	\maketitle
} \fi

\if0\blind
{
	\bigskip
	\bigskip
	\bigskip
	\title{\bf\large Approximate Factor Models with Strongly Correlated Idiosyncratic Errors}
	\date{\vspace{-5ex}}
	\maketitle
	\medskip
} \fi
\thispagestyle{empty}
\bigskip
\begin{abstract}
	We consider the estimation of approximate factor models for time series data, where strong serial and cross-sectional correlations amongst the idiosyncratic component are present. This setting comes up naturally in many applications, but existing approaches in the literature rely on the assumption that such correlations are weak, leading to mis-specification of the number of factors selected and consequently inaccurate inference. In this paper, we explicitly incorporate the dependent structure present in the idiosyncratic component through lagged values of the observed multivariate time series. We formulate a constrained optimization problem to estimate the factor space and the transition matrices of the lagged values {\em simultaneously}, wherein the constraints reflect the low rank nature of the common factors and the sparsity of the transition matrices. We establish theoretical properties of the obtained estimates, and introduce an easy-to-implement computational procedure for empirical work. The performance of the model and the implementation procedure is evaluated on synthetic data and compared with competing approaches, and further illustrated on a data set involving weekly log-returns of 75 US large financial institutions for the 2001--2016 period.
\end{abstract}

\noindent%
\vfill

\clearpage
\setcounter{page}{1}
\spacingset{1} % DON'T change the spacing!

%%%%%%%%%%%%% INTRO %%%%%%%%%%%%%%%%%%
\section{Introduction.}\label{sec:intro}

Factor models are widely used in a number of scientific fields for reducing the dimension of data sets comprising of a large number of variables \citep{anderson2003introduction}. A factor model assumes that each variable under consideration can be expressed as a linear combination of a small number of {\em latent} factors plus an idiosyncratic component 
(error term). Co-movements among the variables can be accounted for by these few factors thus aiding interpretation. When {\em exact} factor models are used in the analysis of cross-sectional data, it is assumed that the idiosyncratic components are mutually uncorrelated \citep{anderson2003introduction}. However, for time series data such assumptions are often too restrictive, especially if a large panel of time series is considered where the common factors may not fully capture all relationships among the observed time series; in this case, it is of interest to examine {\em approximate} factor models that also allow for correlations amongst the idiosyncratic components.

Such an approximate factor model was introduced in \citet{chamberlain1983arbitrage} for the analysis of portfolios comprising of a large number of assets. Since then, a number of papers have appeared in the literature investigating properties of such approximate factor models, under the assumption that the correlations between the common factors and the idiosyncratic component, as well as those amongst the idiosyncratic components are {\em weak}. Formally, the approximate factor model is defined as 
\begin{equation}\label{approximate_factor}
X_t = \Lambda F_t + u_t, \ \ t=1,\cdots,T,
\end{equation}
where $X_t$ is a vector of $p$-dimensional time series, $F_t$ a $K$-dimensional latent factor process, $\Lambda$ a $p\times K$ matrix of {\em factor loadings} and $u_t$ the vector of idiosyncratic components. It is often further assumed that the factor process exhibits Vector Autoregressive dynamics, namely $F_t = \sum_{i=1}^{q} \Phi_i F_{t-i}+\eta_t$, where $\eta_t$ is a serially uncorrelated error process that is independent across its coordinates, and $\Phi_i$ are $K\times K$ transition matrices. The model in~\eqref{approximate_factor} is typically estimated through principal component (PC) decomposition, which operates under the assumption that as the time series panel size $p\rightarrow\infty$, the leading $K$ eigenvalues of $\Sigma_X:=\mathbb{E}(X_tX_t^\top)$ diverge, whereas all eigenvalues of $\Sigma_u:=\mathbb{E}(u_tu_t^\top)$ are bounded, thus enabling the separation between the common factors and the idiosyncratic components. Some key theoretical results for this model are given in \citet{bai2002determining,bai2003inferential}, where asymptotic normality of the estimated factors and factor loadings\footnote{up to some invertible transformation} obtained from PC analysis is established, under a $\sqrt{p}/T\rightarrow 0$ scaling for the former result and a $\sqrt{T}/p\rightarrow 0$ scaling for the latter. Further, if $T/p\rightarrow 0$, then the maximum time-indexed deviation of the estimated factors relative to their true values vanishes. In later work, \citet{stock2005implications} consider the same factor model representation, but each coordinate of $u_t$ is allowed to exhibit serial correlation, and assumed to be uncorrelated with $F_t$ across all time leads and lags. By decorrelating the coordinates of the $u_t$ error process\footnote{With a slight abuse of notion, here we use $F_t$ to denote the term that collects the lags of the actual factors that enter the model after decorrelation.}, the model can written in the form of 
\begin{equation}\label{SW-formulation}
X_t = \Lambda F_t + D(L)X_{t-1} + \epsilon_t,
\end{equation}
where $D(L)=\text{diag}(\delta_1(L),\dots,\delta_p(L))$ is a {\em diagonal} matrix with each entry being the autoregressive polynomial corresponding to coordinates of $u_{t}$, while $\epsilon_t$ is a pure noise term that is neither cross-sectionally nor serially correlated. 

The presence of strongly correlated idiosyncratic components in the model can lead to distorted estimation and inference, resulting in overestimation of the number of factors \citep{greenaway2012estimating} and is detrimental for forecasting purposes \citep{anderson2007forecasting}. Within the DFM framework, a common remedy entails the inclusion of lagged terms of $X_t$ \citep[e.g.][]{anderson2007forecasting,carare2010spillovers,liu2013modelling,eichengreen2012subprime}, which however augments the number of model parameters at the rate of $p^2$. To overcome the technical issues arising from jointly estimating a large number of parameters, these methods either implicitly assume a small panel size \citep[e.g.][]{anderson2007forecasting} so that their application does not suffer from the curse of dimensionality, or resort to estimating $p\choose2$ models based on pairwise univariate series from the $X_t$ components \citep[e.g.][]{liu2013modelling}. 

In a related line of work, \citet{forni2000generalized} introduced the generalized dynamic factor model (GDFM) framework that dictates the existence of two {\em mutually orthogonal} processes that capture the common and idiosyncratic components, respectively. The dynamic factors spanning the common space can be general \mbox{$L^2$-integrable} processes and estimated through principal components in the frequency domain \citep{forni2005generalized,forni2015dynamic}.

Despite the generality of the GDFM framework, whose formulation ensures orthogonality between the independent and identically distributed noise process and the common space, its common space recovery relies on estimated spectral density matrices that can
exhibit numerical instabilities when the dimensionality of $X_t$ becomes large \citep{fiecas2014data}. On the other hand, the aforementioned shortcomings of the DFM framework can be largely mitigated, through the inclusion of lagged terms and a formulation that {\em jointly} estimates model parameters via a computationally stable procedure.
To this end, following \citet{stock2005implications}, we write the approximate factor model in the form given in~\eqref{SW-formulation}, but allow for $D(L)$ to exhibit cross-correlation structure; i.e. $D(L)$ is not restricted to be diagonal, but merely {\em sparse}. Hence, the dynamics of the $p$ time series in $X_t$ can be written in the form of a lag-adjusted static factor model, with the lag term impacting the current values through {\em sparse} transition matrices. 
Through cross-sectional de-correlation, $\epsilon_t$ becomes a strictly exogenous noise process comprising of independent and identically distributed shocks, and the model representation aligns with that under the GDFM framework, with the lagged term(s) and the factors collectively capturing the common space and accounting for all pervasive shocks.
In the proposed model specification, two key quantities are the space spanned by the common factors (the ``factor hyperplane" henceforth), and the sparse transition matrices of the time lags of the observable process. To obtain their estimates, we formulate a {\em penalized maximum likelihood} objective function, introduce a block-coordinate descent algorithm to solve the posited optimization problem and establish finite sample high-probability error bounds for the convergent solution estimates. Finally, note that the transition matrices of the lagged values can provide useful and interpretable information, as shown in our application study and noted in \citet{eichengreen2012subprime,liu2013modelling}.

Together with the proposed model specification that allows for strongly correlated idiosyncratic components in approximate factor models, key contributions of this work entail the convex formulation that leads to joint estimation of the model parameters, as well as the technical developments that provide insights on appropriately handling the interaction between the {\em latent} factor space and the lagged space spanned by the past history of the observed process. In particular, the strategy used to establish error bounds is applicable to other high-dimensional statistical models involving simultaneously observed and latent components. 

The remainder of this paper is organized as follows. In Section~\ref{sec:problem}, we introduce our model setup, estimation procedure for model parameters, as well as steps for performing forecast. Theoretical properties of the proposed estimators are established in Section~\ref{sec:theory}, including their high-probability statistical error bound, and convergence property.
In Section~\ref{sec:simulation}, we introduce an empirical implementation procedure and present the performance evaluation of the estimates based on synthetic data. In Section~\ref{sec:realdata}, an application of our model to weekly stock return data of large US financial institutions for the period from 2001 to 2017 period is considered. Finally, Section~\ref{sec:discussion} concludes the paper. 

\noindent 
{\em Notation.} Throughout this paper, for some generic matrix $A$ of dimension $m\times n$, we use $\smalliii{\cdot}$ to denote its matrix norms, including the operator norm $\smalliii{A}_{\op}$, the Frobenius norm $\smalliii{A}_{\F}$, the nuclear norm $\smalliii{A}_{*}$, $\smalliii{A}_1 = \max_{1\leq j\leq n}\sum_{i=1}^m |a_{ij}|$, and $\smalliii{A}_\infty = \max_{1\leq i\leq m}\sum_{j=1}^n|a_{ij}|$. We use $\|A\|_1=\sum_{i,j}|a_{ij}|$ and $\|A\|_\infty = \max_{i,j}|a_{ij}|$ to denote the elementwise $1$-norm and infinity norm. Additionally, we use $\varrho(A)$ to denote its spectral radius $(\max|\lambda(A)|)$. For two matrices $A$ and $B$ of commensurate dimensions, denote their inner product by $\llangle A, B\rrangle = \tr(A^\top B)$. Finally, we write $A\succsim B$ if there exists some absolute constant $c$ that is independent of the model parameters such that $A\geq cB$.

%%%%%%%%%%%% Formulation %%%%%%%%%%%%%%%%%%
\section{Problem Formulation, Estimation and Forecast.}\label{sec:problem}

We start by introducing the model assuming that the idiosyncratic component follows the aforementioned sparse $\VAR(d)$ model, which simultaneously incorporates the cross-sectional and serial structure among its coordinates. To convey the main arguments, we assume without loss of generality that $d=1$ for the ease of exposition, and present the extension to the general lag case in the Supplement.

The starting point is the dynamic factor representation of the observable process $X_t = \tilde{\lambda}(L)f_t + u_t$, where $f_t$ is the common latent factor; $u_t$ is the idiosyncratic component whose dynamics satisfy $\mathcal{B}(L)u_t = \epsilon_t$ with $\mathcal{B}(L)=\I_p-BL$ being the lagged matrix polynomial for some {\em weakly sparse} $B$. Multiplying $\mathcal{B}(L)$ on both sides leads to the dynamic factor model consisting of~\eqref{model:ours} and~\eqref{model:FVAR}, where \mbox{$F_t\in\mathbb{R}^K (K\ll p)$} collects the lags of $f_t$ so that it only enters the dynamics of $X_t$ contemporaneously, and is additionally assumed to follow some VAR model with lagged polynomial $\Phi(L)$:
\begin{align}
X_t & = \Lambda F_t + BX_{t-1} + \epsilon_t,\label{model:ours} \\
\Phi(L)F_t & = \eta_t.  \label{model:FVAR}
\end{align}
$\epsilon_t$ is a mean zero noise process that is both serially and cross-sectionally uncorrelated. Moreover, it is strictly exogenous satisfying $\Cov(X_{t-1},\epsilon_{t+h})=0$ and $\Cov(F_t,\epsilon_{t+h})=0$, $\forall\,h\geq 0$. The parameters of interest are the factor hyperplane (to be specified later) and the sparse transition matrix $B$. Note that here we only require $B$ to be {\em weakly sparse}, the notion of which can be formalized through the definition of an $\ell_q$ ball with radius $R_q$ \citep[c.f.][]{negahban2012unified}: 
\begin{equation}\label{eqn:Lqball}
\mathbb{B}_q(R_q) := \Big\{ B\in\mathbb{R}^{p\times p}\,\big|\, \sum_{i,j}^p |B_{ij}|^q \leq R_q \Big\} \qquad \text{for some fixed $q\in[0,1]$}.
\end{equation}
The case of exact sparsity corresponds to $q=0$ where $B\in\mathbb{B}_q(R_0)$ has at most $R_0$ nonzero entries; whereas for $q\in(0,1]$, the $R_q$ ball imposes constraints on the decay rate of $|B_{ij}|$'s.

To ensure that $X_t$ is covariance stationary, we require that the spectral radius of $B$ satisfies $\varrho(B)<1$ without further restricting $\Lambda$. Additionally, note that under the assumption that the spectral density of $X_t$ exists, the spectral density of the filtered process $Z_t := \mathcal{B}(L) X_t = \Lambda F_t + \epsilon$ satisfies
\begin{equation*}
g_Z(\omega) = \Lambda g_F(\omega) \lambda^\top + g_\epsilon(\omega) + g_{\epsilon,F}(\omega)\Lambda^\top + \Lambda g_{F,\epsilon}(\omega).
\end{equation*}
Correspondingly, the spectral density of $X_t$ is given by 
\begin{equation*}
g_X(\omega) = \big[ \mathcal{B}^{-1}(e^{i\omega})\big]\Big(\Lambda g_F(\omega)\Lambda^\top + g_\epsilon(\omega) + g_{\epsilon,F}(\omega)\Lambda^\top + \Lambda g_{F,\epsilon}(\omega)\Big)\big[ \mathcal{B}^{-1}(e^{i\omega})\big]^*,
\end{equation*}
where $g_X(\omega)$ and $g_{X,Y}(\omega)$ respectively denote the spectrum and cross-spectrum of some generic process $\{X_t\}$ and $\{Y_t\}$:
\begin{equation*}
g_X(\omega) : = \frac{1}{2\pi}\sum_{h=-\infty}^\infty \Gamma_X(h)e^{-i\omega h} \quad \text{and} \quad 
g_{X,Y}(\omega) := \frac{1}{2\pi} \sum_{h=-\infty}^\infty \Gamma_{X,Y}(h)e^{-i\omega h}
\end{equation*}
with $\Gamma_X(h) :=\E(X_tX_{t-h}^\top)$ and $\Gamma_{X,Y}(h) :=\E(X_tY_{t-h}^\top)$. 

% $$$$$$$$$$$$$$$$$
% MODEL ESTIMTION
% $$$$$$$$$$$$$$$$$
\subsection{Estimation through a convex program.} \label{sec:estimation}

Given a sample of the $p$-dimensional observable process $X_t$, denoted by $\{x_0,x_1,\dots,x_T\}$, let 
\begin{equation*}
\mathbf{X}_{T} := \Big[x_1 ~ x_2 ~\dots ~x_{T}\Big]^\top, \quad \mathbf{X}_{T-1} := \Big[x_0 ~ x_1 ~~ \dots ~ x_{T-1}\Big]^\top, \quad \mathbf{F} :=\Big[ F_1 ~ F_2~ \dots ~ F_T \Big]^\top,
\end{equation*}
where $\mathbf{X}_{T}\in\R^{T\times p}$ and $\mathbf{X}_{T-1}\in\R^{T\times p}$ respectively denote the contemporaneous response matrix and the lagged predictor matrix, and $\mathbf{F}\in\R^{T\times K}$ denotes the latent factor matrix with the latent factor $F_t$ at time point $t$ stacked in its rows. The noise matrix $\mathbf{E}$ is analogously defined. We additionally define the {\em factor hyperplane} associated with the latent factor $\mathbf{F}$ as $\Theta:=\mathbf{F}\Lambda^\top \in \R^{T\times p}$, and note that $\Theta$ has rank at most $K$. With the above notations, the model in~\eqref{model:ours} for the observed samples can be written as $\mathbf{X}_{T} = \Theta + \mathbf{X}_{T-1}B^\top + \mathbf{E}$. With the transition matrix $B\in\R^{p\times p}$ assumed sparse and the factor hyperplane $\Theta=\mathbf{F}\Lambda^\top\in\R^{T\times p}$ being low rank, we formulate the following constrained optimization problem:
\begin{equation}\label{eqn:optimization_v0}
\min\limits_{B,\Theta}~\Big\{ \tfrac{1}{2T}\smalliii{ \mathbf{X}_T-\Theta-\mathbf{X}_{T-1}B^\top}_{\F}^2\Big\}, ~~~\text{subject to}~\text{rank}(\Theta)\leq r,~\|B\|_1\leq L,~\text{for some $r$ and $L$},
\end{equation}
with the feasible region determined through a rank constraint imposed on $\Theta$ and a sparsity-inducing norm constraint imposed on $B$. 

The rank constraint in~\eqref{eqn:optimization_v0} leads to a {\em non-convex} feasible region, making it particularly hard to characterize the obtained solution analytically as it depends on the initial values provided to the algorithm. Thus, as commonly undertaken in the literature \citep[e.g.][]{agarwal2012noisy}, we consider a tight convex relaxation of the rank constraint, and the solution to the convexified program has convergence guarantees independent of the initializer. Formally, we consider obtaining the estimator through the convex program in~\eqref{eqn:optimization0_hetero}, which can be obtained from~\eqref{eqn:optimization_v0} by alternatively considering the nuclear norm constraint for the factor hyperplane and the $\ell_1$ norm constraint for the sparse transition matrix $B$ in Lagrangian form:
{\small
	\begin{equation}\label{eqn:optimization0_hetero}
	(\widehat{B},\widehat{\Theta})  =  \argmin\limits_{B,\Theta} \Big\{ \frac{1}{2T}\smalliii{ \mathbf{X}_T-\Theta-\mathbf{X}_{T-1}B^\top}_{\F}^2
	+ \lambda_B \vertii{B}_1 + \lambda_\Theta \smalliii{\Theta/\sqrt{T}}_* \Big\}, 
	\end{equation}}%
where $\lambda_B$ and $\lambda_\Theta$ are tuning parameters. The solution $(\widehat{B},\widehat{\Theta})$ can be obtained by a block-coordinate descent algorithm which alternately minimizes with respect to $B$ and $\Theta$, as outlined in Algorithm~\ref{algo:1}. 

% $$$$$$$$$$$$$$$$$
% ALGO FOR ESTIMATION
\begin{algorithm}[!ht]
	\scriptsize
	\setstretch{0.5}
	\caption{\small Algorithm for estimating $B$ and $\Theta$ by solving~\eqref{eqn:optimization0_hetero}.}\label{algo:1}
	%\SetKwInOut{Input}{\scshape input}
	%\SetKwInOut{Output}{\scshape output}
	\KwIn{Time series data $\{x_t\}_{t=0}^T$, tuning parameter $\lambda_B$, $\lambda_\Theta$}
	\BlankLine
	\textbf{Initialization: set $B^{(0)}=O$}\;
	\textbf{Iterate until convergence}: \vspace*{2mm}\\
	\hspace*{3mm}$\bullet$~~\begin{minipage}[t]{14.4cm}
		For fixed $\widehat{B}^{(m-1)}$, update $\widehat{\Theta}^{(m)}$ by solving$^*$:
		\begin{equation*}
		\widehat{\Theta}^{(m)} =\argmin\nolimits_{\Theta} \big\{ \tfrac{1}{2T}\smalliii{\mathbf{X}_T - \mathbf{X}_{T-1}(\widehat{B}^{(m-1)})^\top -\Theta}_{\F}^2 + \lambda_\Theta \smalliii{\Theta/\sqrt{T}}_* \big\}.
		\end{equation*}
	\end{minipage}\vspace*{5mm}\\
	\hspace*{3mm}$\bullet$~~\begin{minipage}[t]{14.4cm}
		For fixed $\widehat{\Theta}^{(m)}$, update $\widehat{B}^{(m)}$ with each row $j$ (in parallel) solving a Lasso regression 
		\begin{equation*}
		\widehat{B}^{(m)}_{j\cdot} = \argmin\nolimits_{\beta\in\R^{p}} \big\{ \tfrac{1}{2T}\smallii{\big[\mathbf{X}_{T} - \widehat{\Theta}^{(m)}\big]_{\cdot j} - \mathbf{X}_{T-1}\beta}^2  + \lambda_B \|\beta\|_1  \big\}.
		\end{equation*}
	\end{minipage}\vspace*{2mm}\\
	\KwOut{Estimated sparse transition matrix $\widehat{B}$ and the low rank hyperplane $\widehat{\Theta}$.}\vspace*{3mm}
	\scriptsize
	\noindent *Note: $\Theta$-update can be obtained by a proximal gradient descent algorithm involving a soft singular value thresholding (SVT) step, with each inner iteration indexed by $t$ solving the following: $$\widehat{\Theta}^{(m,[t+1])} = \argmin\big\{ \llangle \Theta, \nabla \mathcal{G}^{(m)}(\widehat{\Theta}^{(m,[t])})\rrangle + \frac{\eta}{2} \smalliii{\Theta-\Theta^{(m,[t])}}_{\F}^2 + \lambda_\Theta\smalliii{\Theta/\sqrt{T}}_* \big\}~~\text{for some stepsize $\eta$},$$
	where $\mathcal{G}^{(m)}(\Theta) := \tfrac{1}{2T}\smalliii{\mathbf{X}_T - \mathbf{X}_{T-1}(\widehat{B}^{(m-1)})^\top -\Theta}_{\F}^2$.
\end{algorithm}
% END OF ALGO
% $$$$$$$$$$$$$$$$$

\normalsize
\paragraph{Reconstruction of the factors.} The solution to~\eqref{eqn:optimization0} provides an estimate of the factor hyperplane, based on which realizations of the $K$-dimensional latent factors process can be reconstructed under certain identifiability restrictions. As mentioned in Section~\ref{sec:intro}, for any invertible matrix $R\in\R^{K\times K}$, the following equality holds
\begin{equation*}
\Theta = \mathbf{F}\Lambda^\top = \big[\mathbf{F}R^\top \big] \big[ \Lambda R^{-1} \big] ^\top :=\check{\mathbf{F}}\check{\Lambda}^\top,
\end{equation*}
hence, given a factor hyperplane and the latency of the factors, to fully identify the factors and the corresponding loading matrix $(\mathbf{F},\Lambda)$ from their observationally equivalent counterpart $(\check{\mathbf{F}},\check{\Lambda})$, a total number of $K^2$ restrictions is required to address their indeterminacy. Various choices for the identification restrictions have been discussed in the literature \citep[e.g.,][and references therein]{bai2008large}, including the most popular PC estimator \citep{stock2002forecasting} which assumes orthogonality for both the factors and the loadings, as well as the ones that implicitly assume certain ordering of the factors and impose specific structural restrictions on the loading matrix \citep[see PC2 and PC3 identification restrictions in][]{bai2013principal}. Under these restrictions, the factors and the loading matrix can always be uniquely identified\footnote{For the PC estimator or under the PC2 restriction, where $\mathbf{F}'\mathbf{F}/n=\I_K$ and $\Lambda$ is assumed lower-triangular, the identification is up to sign rotation; under the PC3 one, where the upper $K\times K$ upper sub-matrix of $\Lambda$ is assumed an identity matrix and $\mathbf{F}$ is left unrestricted, the identification is exact \citep[see][]{bai2013principal}.} and obtained based on the SVD of the estimated hyperplane $\widehat{\Theta}$. It is worth noting that regardless of the identification restrictions that lead to different versions of the estimated factors, the space spanned by the estimated factors is invariant once $\widehat{\Theta}$ is obtained. Specifically, forecasting future values of $x_t$ does not require an exact recovery of $F_t$, as discussed next.

\subsection{Forecasting.}
Given estimates $\widehat{B}$ of the transition matrix and $\widehat{\Theta}$ of the hyperplane, we consider the following procedure that first obtains forecasts of the filtered process $Z_t:=X_t - BX_{t-1}$ through projection onto the factor space, followed by a lag adjustment to obtain those of the $X_t$. 

To this end, according to the model in~\eqref{model:ours}, the filtered process $Z_t$ can be represented as $Z_t = \Lambda F_t + \epsilon_t$, whose $h$-step-ahead best linear predictor based on $F_{T-k},k\geq 0$ is given by the projection $\text{Proj}(Z_{T+h}\,|\,\text{Span}(\mathbf{F},T))$, 
where $\text{Span}(\mathbf{F},T)$ denotes the linear space spanned by $\{F_t\}_{t=1}^T$ \citep{stock2002forecasting,forni2005generalized}. In particular, based on estimate $\widehat{B}$, the filtered process $Z_t$ can be estimated through $\widehat{z}_t:= x_t - \widehat{B}x_{t-1}$, whose common space estimate corresponds to $\widehat{\Theta}$.
Using the surrogate process $\{\widehat{z}_t\}$, let the sample covariance be $\widehat{\Gamma}_Z(h):=\tfrac{1}{T-h}\sum_{t=h+1}^T \widehat{z}_t \widehat{z}^\top_{t-h}$; the $h$-step-ahead forecast of $\{z_{t}\}$ is then given by
\begin{equation}
\widehat{z}_{T+h} = \big[\widehat{\Gamma}_Z(h) \widehat{V}(\widehat{V}^\top\widehat{\Gamma}_Z(0)\widehat{V})^{-1}\big](\widehat{V}^\top\widehat{z}_{T}),
\end{equation}
where columns of $\widehat{V}$ are the right singular vectors of $\widehat{\Theta}$ corresponding to nonzero singular values, and are effectively an orthonormal basis for the factor space. In the case where $h=1$, $\widehat{x}_{T+1} = Bx_T + \widehat{z}_{T+1|T}$; in the case where $h>1$, $\widehat{x}_{T+h}$ can be obtained inductively by sequentially estimating $x_{T+i}$, for all $0<i\leq h$. Algorithm~\ref{algo:forecast} outlines the forecasting procedure. 

% ========================================================
% ALGO FOR FORECASTING
\begin{algorithm}[!ht]
	\scriptsize
	\setstretch{0.5}
	\caption{\small Algorithm for obtaining $h$-step-ahead forecast given $\widehat{B}$ and $\widehat{\Theta}$.}\label{algo:forecast}
	%\SetKwInOut{Input}{\scshape input}
	%\SetKwInOut{Output}{\scshape output}
	\KwIn{Time series data $\{x_t\}_{t=0}^T$, estimates $\widehat{B}$ and $\widehat{\Theta}$}
	\BlankLine
	$\bullet$ Obtain the filtered process $\widehat{z}_t:= \widehat{x}_t - \widehat{B}x_{t-1}$ and its sample cross-covariance estimates $\widehat{\Gamma}_z(i)$, for all $0\leq i\leq h$.\vspace*{2mm}\;
	$\bullet$ For $i=1,2,\cdots,h$: \vspace*{2mm}\\
	\hspace*{3mm}---~~\begin{minipage}[t]{14.4cm}
		Obtain $\widehat{z}_{T+i}$ through $\widehat{z}_{T+i} := \big[\widehat{\Gamma}_Z(i) \widehat{V}(\widehat{V}^\top\widehat{\Gamma}_Z(0)\widehat{V})^{-1}\big](\widehat{V}^\top\widehat{z}_{T})$;
	\end{minipage}\vspace*{2mm}\\
	\hspace*{3mm}---~~\begin{minipage}[t]{14.4cm}
		Obtain $\widehat{x}_{T+i}$ through $\widehat{x}_{T+i}:= \widehat{z}_{T+i} + B\widehat{x}_{T+i-1}$.
	\end{minipage}\vspace*{2mm}\\
	\KwOut{Forecasted values $\widehat{x}_{T+i},0<i\leq h$.}
\end{algorithm}
% END OF ALGO
% =========================================================

\paragraph{Connections to GDFM.} To conclude this section, we discuss similarities of the proposed formulation to the GDFM \citep[e.g.,][]{forni2000generalized,forni2005generalized}. GDFM encompasses a broader class of factor models wherein the observed process admits a decomposition into two mutually orthogonal processes that respectively capture the common and the idiosyncratic component \citep{forni2000generalized}, with the former not limited to a VAR representation. From a modeling perspective, as pointed out in \citet{lutkepohl2014structural}, the distinction between GDFM and the state-space form of DFM \citep[e.g.,][which is also the factor model specification adopted in this paper]{stock2005implications} is not that substantial, since stationary processes can be approximated arbitrarily well by VAR processes with unrestricted order of lags.  From the estimation perspective, however, to accommodate the potentially more complex dynamics of the factor processes, GDFM recovers the common space leveraging the spectral domain features of the processes, and the dynamic factors are obtained by solving a generalized eigen-equation w.r.t. estimated covariance matrices of the common and the idiosyncratic components \citep{forni2005generalized}. In our formulation, we posit an optimization that jointly estimates the factor space and the lagged space, which accounts for the lag information explicitly, under mild sparsity assumptions.

The empirical performance of the two procedures is considered and compared in Section~\ref{sec:simulation} under various data generating mechanisms.

%%%%%%%%%%%%%%%%%% Theory %%%%%%%%%%%%%%%%%%%%
\section{Theoretical Properties.}\label{sec:theory}

To establish statistical properties of the estimators, a ball constraint on the feasible region of $\Theta$ is required to incur additional compactness on the low rank component that limits the spikiness of its entries, and this enables identification of the sparse component $B$. To this end, throughout this section, we consider estimators that are solutions to the following convex program:
\begin{align}\label{eqn:optimization0}
\begin{split}
(\widehat{B},\widehat{\Theta})  = & \argmin\limits_{B,\Theta} \Big\{ \tfrac{1}{2T}\smalliii{\mathbf{X}_T - \Theta - \mathbf{X}_{T-1}B^\top}_{\F}^2  + \lambda_B \vertii{B}_1 + \lambda_\Theta \smalliii{\Theta/\sqrt{T}}_* \Big\}, \\
& \text{subject to}\quad \Theta\in\mathbb{B}_\infty(\phi,\mathbf{X}_{T-1}),
\end{split}
\end{align}
where $\mathbb{B}_\infty(\phi,\mathbf{X}_{T-1})$ is a box constraint given by
\begin{equation*}
\mathbb{B}_\infty(\phi,\mathbf{X}_{T-1}) := \big\{ \Theta: \|\Theta\|_\infty \leq \frac{\phi}{\sqrt{Tp} \cdot \smalliii{\mathbf{X}_{T-1}/\sqrt{T} }}_{\text{op}} \big\}.
\end{equation*}
$\phi$ is chosen such that the true value of the parameters $\Theta^\star$ is always feasible. We will provide further illustration on the interpretation of such a box constraint in Section~\ref{sec:theory-fix} and Remark~\ref{remark:boxconstraint}. $(\widehat{B},\widehat{\Theta})$ falls into the class of {\em regularized $M$-estimators}, whose theoretical properties have been extensively studied in the statistical literature for diverse settings \citep[e.g.,][]{agarwal2012fast,loh2012high}. 

A road map to establish properties of the estimators for $(\Theta, B)$ is given next: first in Section~\ref{sec:theory-fix} we derive non-asymptotic statistical error bounds of $\widehat{\Theta}$ and $\widehat{B}$ under certain regularity conditions, when the proposed estimation procedure is based on a {\em deterministic realization} of the observable process $\{X_t\}$. In particular, the required regularity conditions primarily entail the {\em restricted strong convexity} (RSC) condition \citep{agarwal2012noisy} and that the choice of $\lambda_B$ and $\lambda_\Theta$ is in accordance with some {\em deviation condition} \citep{loh2012high}. Subsequently, in Section~\ref{sec:theory-random}, we establish that the required conditions are satisfied with high probability, and provide probabilistic analogues of key model parameters' error bounds for {\em random realizations} drawn from the underlying observable Gaussian process $\{X_t\}$ and the latent process $\{F_t\}$. We also briefly discuss how the model identifiability issue is tackled through the constrained formulation adopted in~\eqref{eqn:optimization0}. Finally in Section~\ref{sec:convergence}, from a numerical perspective, we establish the convergence of the proposed iterative algorithm to a stationary point. All proofs are deferred to Appendices~\ref{appendix:proof-of-theory-convex} and~\ref{appendix:auxLemma}. Throughout our exposition, we use superscript $\star$ to denote the true value of the parameters of interest, and denote the errors of the estimators by $\Delta_{\Theta}:=\widehat{\Theta}-\Theta^\star$ and $\Delta_B:=\widehat{B}-B^\star$, respectively. 

%%%%%%%%%%%%%%%%%%%%%%%%%%%%%%%%%%%%%%%%%%%%%%%%%%%%%%
%*****************************************************
% THEORY FOR DETERMINISTIC DESIGN
%*****************************************************
%%%%%%%%%%%%%%%%%%%%%%%%%%%%%%%%%%%%%%%%%%%%%%%%%%%%%%

\subsection{Statistical Error Bounds with Deterministic Realizations.}\label{sec:theory-fix}

We start by introducing some additional notation needed in the ensuing technical developments. Let $\ell_T(B,\Theta;\mathbf{X})$ denote the loss function, given by
\begin{equation*}
\ell_T(B,\Theta;\mathbf{X}) := \frac{1}{2T}\smalliii{\mathbf{X}_T - \Theta - \mathbf{X}_{T-1}B^\top}_{\F}^2.
\end{equation*}
The true number of latent factors is given by $K$ and thus $\rank(\Theta^\star)=K$. Further, given some $\eta>0$ (to be specified later), let $S^\star_\eta$ denote the thresholded support set of $B^\star$, and we use $s_\eta$ to denote its cardinality, that is, $S^\star_\eta := \{(i,j)\,|\, |B^\star_{ij}|>\eta \}$ and $s_\eta := \|S^\star_\eta\|_0$. Finally, let $S_\mathbf{E} : = \mathbf{E}^\top_{T-1}\mathbf{E}_{n-1}/T$ denote the sample covariance matrix of the error process and let $\Lambda_{\max}(S_\mathbf{E})$ be its maximum eigenvalue. 
%% DEFINITION OF RSC %%
Formally, the {\em RSC condition} \citep[c.f.][]{agarwal2012noisy,negahban2012unified} is defined as follows. 
\begin{definition}[Restricted Strong Convexity (RSC)]\label{defn:RSC} For some generic data matrix $\mathbf{X}\in\R^{T\times p}$, it satisfies the RSC condition with respect to norm $\Phi$ with curvature $\alpha_{\text{RSC}}>0$ and tolerance $\tau_T\geq 0$~if
	\begin{equation*}
	\frac{1}{2T} \smalliii{\mathbf{X}\Delta}_{\F}^2 \geq \frac{\alpha_{\text{RSC}}}{2}\smalliii{\Delta}_{\F}^2 - \tau_T \Phi^2(\Delta), \qquad \forall~\Delta\in\R^{p\times p}.
	\end{equation*}	
	In our context, we consider the element-wise $\ell_1$ norm $\Phi(\Delta)=\|\Delta\|_1$. 
\end{definition}
%% ABOUT DEVIATION %%

Further, for high dimensional sparse VAR models ($\Theta=0$ in the current setup), the tuning parameter $\lambda_B$ needs to satisfy a {\em deviation condition} \citep{loh2012high,basu2015estimation}, namely,
\begin{equation*}
\lambda_B\geq c_0\|\nabla_B\ell_T - \nabla^2_B\ell_T\cdot (B^\star)^\top \|_\infty,\qquad \text{for some constant}~~c_0>0,
\end{equation*}
which can be simplified to $\lambda_B\geq \|\mathbf{X}_{T-1}^\top\mathbf{E}/T\|_\infty$. Under the current model setup, however, the deviation condition is significantly more involved and requires proper modifications to incorporate quantities associated with the factor hyperplane, as seen in Theorem~\ref{thm:fix-bound-convex}.

Before stating the main results, we provide a brief discussion on the box constraint on $\Theta$, which aims to ``limit" the spikiness of the low rank component, and hence the interaction between the latent factor space and the observable lag space spanned by $X_{t-1}$ --- in particular, for $\Theta$ and $B$ to be properly recovered, such interaction can not be too large. Due to the basis vectors of the factor space being latent, a direct restriction on the interaction is impractical and conceptually unsatisfying, whereas the box constraint adopted effectively restricts the product of the signals from the two spaces and serves our objective, as shown in the proof of Theorem~\ref{thm:fix-bound-convex} and Remark~\ref{remark:id}. Note that this constraint is in the same spirit to similar ones in the literature \citep[e.g.,][]{agarwal2012noisy,negahban2012restricted}, and the norm of $\mathbf{X}$ is necessary since the two spaces have distinct bases. 

%% STATEMENT OF THE MAIN THEOREM %%
\begin{theorem}[Error bound for $(\widehat{B},\widehat{\Theta})$ under fixed realizations]\label{thm:fix-bound-convex} Suppose fixed realizations $\mathbf{X}_{T-1}\in\R^{T\times p}$ of process $X_t\in\R^p$ satisfy the RSC condition with curvature $\alpha_{\text{RSC}}>0$ and a tolerance $\tau_T$ such that 
	\begin{equation}\label{eqn:sparsity}
	64\tau_T\Big(s_\eta + (2K)\big( \frac{\lambda_{\Theta}}{\lambda_B}\big)^2 \Big) < \min\{\alpha_{\text{RSC}},1\}.
	\end{equation}
	Then, for any matrix pair $(B^\star,\Theta^\star)$ that generates the evolution of the $X_t$ process, for estimators $(\widehat{B},\widehat{\Theta})$ obtained by solving the optimization~\eqref{eqn:optimization0} with regularization parameters $\lambda_B$ and $\lambda_{\Theta}$ satisfying
	\begin{equation}\label{eqn:choice_of_lambda}
	\lambda_B \geq 2\|\mathbf{X}_{T-1}^\top \mathbf{E}/T\|_\infty + 4\phi/\sqrt{Tp} \qquad \text{and} \qquad \lambda_\Theta \geq \Lambda^{1/2}_{\max}(S_{\mathbf{E}}),
	\end{equation}
	the following error bound holds for some positive constants $C_1$, $C_2$ and $C_3$:
	\begin{equation}\label{eqn:eb1}
	\smalliii{\Delta_B}_{\F}^2 + \smalliii{\Delta_{\Theta}/\sqrt{T}}_{\F}^2 \leq C_1\cdot \mathcal{E}_{B} + C_2 \cdot \mathcal{E}_{\Theta} + C_3 \cdot \mathcal{E}_{\tau_T},
	\end{equation}
	where $\alpha^\prime := \min\{\alpha_{\text{RSC}},1\}$,
	\begin{equation*}
	\mathcal{E}_{B}:=\Big(\frac{\lambda_B}{\alpha^\prime}\Big)^2 \Big\{ s_\eta + \frac{\alpha^\prime}{\lambda_B} \sum_{(i,j)\notin S^\star_\eta}  |B^\star_{ij}| \Big\},
	\quad \mathcal{E}_{\Theta}:= \Big(\frac{\lambda_\Theta}{\alpha^\prime}\Big)^2 K, 
	\quad \mathcal{E}_{\tau_T}:= \Big(\frac{\tau_T}{\alpha^\prime}\Big)(\sum_{(i,j)\notin S^\star_\eta}  |B^\star_{ij}|)^2.
	\end{equation*}
\end{theorem}

Next, we comment on the error bound in~\eqref{eqn:eb1} and the required conditions in~\eqref{eqn:sparsity}. The error bound encompasses three terms that are respectively associated with the transition matrix $B$, the low rank factor space $\Theta$, and the tolerance $\tau_T$ which measures the extent to which the log-likelihood function deviates from strong convexity (see Definition~\ref{defn:RSC}). Both $\mathcal{E}_B$ and $\mathcal{E}_\Theta$ depend on three components: (1) the overall curvature of the log-likelihood function as captured by $\alpha_{\text{RSC}}$, (2) the interaction structure between various components of the underlying process, as captured by the tuning parameters $\lambda_B$ and $\lambda_\Theta$, and (3) the inherent structure of the parameters as captured by $s_\eta$, $\sum_{(i,j)\notin S^{\star}_{\eta}}|B_{ij}|$ and $K$ --- in particular, due to the approximately sparse structure of $B^\star$, both the density level $s_\eta$ of its strong support set and the magnitude of its ``weak" entries play a role, with the two respectively reflecting the {\em estimation error} and the {\em approximation error} \citep[c.f.][]{agarwal2012noisy} after proper scaling. 
The curvature as measured by $\alpha_{\text{RSC}}$ dictates the constraint to which the tolerance $\tau_T$ needs to conform (see Equation~\eqref{eqn:sparsity}), and such a constraint is also interrelated to $s_\eta$ and $K$: for~\eqref{eqn:sparsity} to be satisfied, neither $K$ nor $s_\eta$ can be too large. Moving to the tuning parameters, $\lambda_B$ can be sub-divided into two terms: the cross-product term $\|\mathbf{X}_{T-1}^\top \mathbf{E}/T\|_\infty$ measures the maximum interaction between the design matrix $\mathbf{X}_{T-1}$ and the noise $\mathbf{E}$, which according to model assumption (population level) should center around 0; where the term $\phi/\sqrt{Tp}$ corresponds to an upper bound on the interaction between the latent (factor) space and the observed one ($\mathbf{X}_{T-1}$). For $\lambda_\Theta$, we require that it dominates the maximum signal coming from the error process in the form of $\Lambda^{1/2}_{\max}(S_{\mathbf{E}})$. Thus, a smaller $\lambda_B$ is needed when interactions between associated terms are weaker and similarly a smaller $\lambda_\Theta$ is needed if the magnitude of the noise is weaker, thus leading to a tighter error bound for the estimates.
Finally, it is worth noting that $\mathcal{E}_{\tau_T}$ is a result of the approximately sparse structure of $B$; in the special case where $B$ is exactly sparse, this term would be 0.

%*****************************************************
% Corollary -- under specific choice of thresholded level eta
%*****************************************************
Corollary~\ref{cor:eb} gives the bound of $\Delta_B$ and $\Delta_\Theta$ with specific choice of the thresholded level $\eta$ , when the true value $B^\star$ lies in the $\ell_q$ ball of radius $R_q$ (see definition in Equation~\eqref{eqn:Lqball}):
\begin{corollary}\label{cor:eb} Under the same set of conditions as in Theorem~\ref{thm:fix-bound-convex}, with $B^\star\in \mathbb{B}_q(R_q)$, by choosing the thresholded level according to $\eta = \lambda_B/\alpha^\prime$ where $\alpha^\prime:=\min\{\alpha_{\text{RSC}},1\}$, the following error bound holds for some positive constants $C_1,C_2$ and $C_3$:
	\begin{equation*}
	\smalliii{\Delta_B}_{\F}^2 + \smalliii{\Delta_{\Theta}/\sqrt{T}}_{\F}^2 \leq C_1\cdot (\frac{\lambda_B}{\alpha^\prime})^{2-q}R_q + C_2\cdot(\frac{\lambda_\Theta}{\alpha^\prime})^2 K + C_3\cdot\frac{\tau_T}{\alpha^\prime}(\frac{\lambda_B}{\alpha^\prime})^{2-q}R_q^2.
	\end{equation*}
\end{corollary}

%*****************************************************
% HIGH PROBABILITY BOUND UNDER RANDOM REALIZATIONS
%*****************************************************
\subsection{High Probability Bounds under Random Realizations.}\label{sec:theory-random}

Next, we provide high probability bounds/concentrations for key quantities associated with the derived error bound in Section~\ref{sec:theory-fix}, for random Gaussian realizations of the underlying factor and error processes. Specifically, this involves the verification of the RSC condition, as well as the examination of quantities associated with the deviation condition to which the choice of $(\lambda_B,\lambda_\Theta)$ needs to conform, as shown in~\eqref{eqn:choice_of_lambda}.

We introduce additional notation for the subsequent technical developments. For some generic process $\{X_t\}$, in addition to the auto-covariance function $\Gamma_X(h)$ and its spectral density $g_X(\omega)$, we define its maximum and minimum eigenvalue associated with the spectral density $g_X(\omega)$ introduced in Section \ref{sec:problem} as follows \citep{basu2015estimation}:
\begin{equation*}
\mathcal{M}(g_X) : = \mathop{\text{ess sup }}\limits_{\omega\in[-\pi,\pi]} \Lambda_{\max}(g_X(\omega)), \qquad \mathfrak{m}(g_X) : = \mathop{\text{ess inf }}\limits_{\omega\in[-\pi,\pi]} \Lambda_{\min}(g_X(\omega)).
\end{equation*}
For two generic centered processes $\{X_t\}$ and $\{Y_t\}$ that are assumed jointly covariance stationary, whose spectral density is given by $g_{X,Y}(\omega):=\tfrac{1}{2\pi}\sum_{h=-\infty}^{\infty} \Gamma_{X,Y}(h)e^{i\omega h}$ where $\Gamma_{X,Y}(h)=\mathbb{E}(X_tY_{t+h}^\top)$, the upper extreme for $g_{X,Y}(\omega)$ is analogously defined as
\begin{equation*}
\mathcal{M}(g_{X,Y}) : = \mathop{\text{ess sup }}\limits_{\omega\in[-\pi,\pi]} \sqrt{ \Lambda_{\max}\big( g^*_{X,Y}(\omega) g_{X,Y}(\omega)\big) }.
\end{equation*}
In general $g_{X,Y}(\omega)\neq g_{Y,X}(\omega)$, but $\mathcal{M}(g_{X,Y})=\mathcal{M}(g_{Y,X})$. 

For the processes involved in our proposed model, recall that $\{X_t\}$, $\{\epsilon_t\}$ and $\{F_t\}$ are mean zero Gaussian processes. In particular, $\{\epsilon_t\}$  is a noise process that does not exhibit temporal nor cross-sectional dependence, hence it is effectively a Gaussian random vector with covariance $\Sigma_\epsilon=\sigma_\epsilon^2 \I_p$, and its spectral density simplifies to $g_\epsilon(\omega)=\frac{\Sigma_\epsilon}{2\pi}$. Further, we define the shifted process $\{\widetilde{\epsilon}_t:=\epsilon_{t+1}\}$ for notation convenience. 

%%%% RSC %%%%%
The following lemma verifies that with high probability, for random realizations of the process $\{X_t\}$, the RSC condition is satisfied provided that the sample size is sufficiently large: 
\begin{lemma}[verification of the RSC condition]\label{lemma:RSC} Consider $\mathbf{X}\in\R^{T\times p}$ whose rows are some random realization $\{x_0,\dots,x_{T-1}\}$ of the stable $\{X_t\}$ process with dynamic given in~\eqref{model:ours}. Then there exist positive constants $c_i~(i=0,1,2,3)$ such that with probability at least $1-c_1\exp(-c_2T)$, the RSC condition holds for $\mathbf{X}$ with curvature $\alpha_{\text{RSC}}$ and tolerance~$\tau_T$ satisfying
	\begin{equation*}
	\alpha_{\text{RSC}} = \pi\mathfrak{m}(g_X),~~~\text{and}~~~\tau_T = \gamma^2\Big(\frac{\alpha_{\text{RSC}}}{2}\Big)\Big( \frac{\log p}{T}\Big)~~~\text{where $\gamma:= 54 \mathcal{M}(g_X)/\mathfrak{m}(g_X)$},
	\end{equation*}
	provided that $T\succsim s_\eta\log p$. 
\end{lemma}

%%%% DEVIATION %%%%%
The next lemma establishes a high probability bound for the interaction term $\mathbf{X}_{T-1}^\top \mathbf{E}/T$ that influences the choice of $\lambda_B$ through its elementwise $\ell_\infty$ norm. 

\begin{lemma}[High probability bound for $\|\mathbf{X}^\top_{T-1}\mathbf{E}/T\|_\infty$]\label{lemma:boundXE} There exist positive constants $c_i~(i=0,1,2)$ such that for sample size $T\succsim \log p$, with probability at least $1-c_1\exp(-c_2\log p)$, the following bound holds:
	\begin{equation}\label{eqn:boundXE}
	\|\mathbf{X}^\top_{T-1}\mathbf{E}/T\|_\infty \leq c_0\Big(\mathcal{M}(g_X) + \mathcal{M}(g_\epsilon) + \mathcal{M}(g_{X,\widetilde{\epsilon}}) \Big)\sqrt{\frac{\log p}{T}}.
	\end{equation}
\end{lemma}

Note that with the definition of the shifted processes $\{\widetilde{\epsilon}_t\}$, we have 
$g_{X,\widetilde{\epsilon}}(\omega) = e^{-ih\omega} g_{X,\epsilon}(\omega)$, which implies $\mathcal{M}(g_{X,\widetilde{\epsilon}})=\mathcal{M}(g_{X,\epsilon})$. Hence, the term that measures the upper extreme of the cross-spectrum between $X_t$ and the shifted process in~\eqref{eqn:boundXE} can be replaced by its unshifted counterpart. Moreover, since $g_\epsilon(\omega)=\tfrac{\sigma_\epsilon}{2\pi}$, its upper extreme is given by $\mathcal{M}(g_\epsilon) = \Lambda_{\max}(\Sigma_\epsilon)/(2\pi)$. 

%%%% MAXIMUM EIGEN-SPECTRA %%%%%
The next lemma provides an upper bound for the maximum eigenvalue of the sample covariance matrix.
\begin{lemma}[High probability concentration for $\Lambda_{\max}(S_{\mathbf{E}})$]\label{lemma:Emax} Consider $\mathbf{E}\in\R^{T\times p}$ whose rows are independent realizations of the mean zero Gaussian random vector $\epsilon_t$ with covariance $\Sigma_\epsilon$. Then, for sample size $T\succsim p$, with probability at least $1-\exp(-T/2)$, the following bound holds:
	\begin{equation*}
	\Lambda_{\max}(S_{\mathbf{E}}) \leq 9\Lambda_{\max}(\Sigma_\epsilon). 
	\end{equation*}
\end{lemma}
Proofs for Lemmas~\ref{lemma:RSC} to~\ref{lemma:Emax} can be found in Appendix~\ref{appendix:proof-of-theory-lemma}.

% $$$$$$$$$$$$$$$$$$$$$$$$$$$$$$$$$$$$$$$$
% HIGH PROBABILITY ERROR BOUND; COMBINED
% $$$$$$$$$$$$$$$$$$$$$$$$$$$$$$$$$$$$$$$$
Up to this stage, we have verified the RSC condition and obtained the high probability bounds for quantities that are associated with the choice of $(\lambda_B,\lambda_\Theta)$, for random realizations from the underlying processes. Theorem~\ref{thm:combine} combines the results in Corollary~\ref{cor:eb} and Lemmas~\ref{lemma:RSC} to~\ref{lemma:Emax}, and provides a high probability error bound of the estimates when the data are random realizations from the underlying processes, as stated next. 

\begin{theorem}[High probability error bound with random realizations]\label{thm:combine} Suppose we are given a snapshot of length $(T+1)$ $\{x_0,\dots,x_{T}\}$ from the $p$-dimensional observable process $\{X_t\}$, whose dynamics are described in~\eqref{model:ours} with $B^\star\in \mathbb{B}_q(R_q)$. Then, there exist universal positive constants $c_i~(i=1,2)$ and $c_i'~(i=1,2)$ such that for sample size $T\succsim p$, by solving convex problem~\eqref{eqn:optimization0} with regularization parameters
	\begin{equation*}
	\lambda_B = c_1\big(\mathcal{M}(g_X) + \mathcal{M}(g_\epsilon) + \mathcal{M}(g_{X,\epsilon})\big)\sqrt{\tfrac{\log p}{T}} + 4\phi/\sqrt{Tp} \quad \text{and} \quad \lambda_\Theta = c_2\Lambda^{1/2}_{\max}(\Sigma_\epsilon),
	\end{equation*}
	the solution $(\widehat{B},\widehat{\Theta})$ has the following bound with probability at least $1-c_1'\exp\big(-c_2'\log p\big)$, by choosing the thresholded level at $\kappa\lambda_B$ with $\kappa:=\max\big\{ \mathfrak{m}^{-1}(g_X),\pi \big\}$:
	\begin{equation}\label{eqn:hpbound}
	\smalliii{\Delta_B}^2_{\F} + \smalliii{\Delta_{\Theta}/\sqrt{T}}^2_{\F} \leq C_1\cdot \kappa^{2-q} \lambda_B^{2-q}R_q + C_2\cdot \kappa^2 K + C_3 \Big( \frac{\log p}{T}\Big)( \kappa \lambda )^{2-2q}R_q^2,
	\end{equation}
	where $C_i~(i=1,2)$ are positive constants that are independent of $T$ and $p$.
\end{theorem}
% $$$$$$$$$$$$$$$$$$$$$$$$$
% REMARK: discussion of results: model identifiability
% $$$$$$$$$$$$$$$$$$$$$$$$$
\begin{remark}\label{remark:id} Note that Theorem~\ref{thm:combine} requires that $T\succsim p$ for relevant quantities to properly concentrate; as a consequence, the estimation errors for $\Delta_B$ and $\Delta_{\Theta}$ are jointly bounded. The sample size requirement is of the same order as in classical factor analysis\footnote{In classical factor analysis, for both the factors and its loadings to be consistently estimated, both $\sqrt{p}/T\rightarrow 0$ and $\sqrt{T}/p\rightarrow 0$ are required to hold simultaneously.} literature \citep[e.g.][]{bai2008large}, and is standard under the context of recovering a low-rank component based on noisy data in high-dimensional statistics \citep[e.g.,][]{agarwal2012noisy}. The nature of the upper bound provided is a consequence of $F_t$ being latent, and hence the low rank factor hyperplane and the error term become not perfectly distinguishable; in particular, the structure of the underlying optimization resembles a noisy matrix completion problem in which the restricted isometry property is violated \citep[see also][]{candes2010matrix}. Further details on model identifiability issues are given in Appendix B.
\end{remark}

% $$$$$$$$$$$$$$$$$$$$$$$$$
% CONVERGENCE ANALYSIS
% $$$$$$$$$$$$$$$$$$$$$$$$$
\subsection{Convergence Analysis of Algorithm \ref{algo:1}.}\label{sec:convergence}

The convergence property of Algorithm~\ref{algo:1} can be established 
using familiar arguments and exploiting its convex nature. Specifically, define the 
objective function is given by
\begin{equation*}
f(B,\Theta):= \ell_T(\mathbf{X}; B,\Theta) + \lambda_B\|B\|_1 + \lambda_\Theta\smalliii{\Theta}_*
\end{equation*}
and is {\em jointly convex} in $(B,\Theta)$, with a convex feasible region $\mathbb{B}_\infty(\phi,\mathbf{X}_{T-1})$. 
Thus, it directly follows from \citet{tseng2001convergence} that the alternating minimization that generates the sequence $\{(\widehat{B}^{(k)},\widehat{\Theta}^{(k)})\}$ converges to a stationary point which is also a global optimum, though the global optimum is not necessarily unique.

To conclude this section, we remark that the theoretical formulation in~\eqref{eqn:optimization0} can be solved in an analogous way to Algorithm~\ref{algo:1}. Specifically, the update of $\Theta$ requires modification to satisfy the constraint on the feasible region of $\Theta$, and the partial minimization can be solved by employing the composite gradient descent algorithm of \citet{nesterov2007gradient} that involves singular value thresholding steps. Nevertheless, the modified algorithm is also convergent, as the one in Algorithm~\ref{algo:1}.

%%%%%%%%%%%%% Implementation %%%%%%%%%%%%%%%
\section{Implementation and Performance Evaluation.}\label{sec:simulation}

In this section, we present results for simulation studies under various settings to demonstrate the performance of our proposed model. As comparison, we also present the common space (to be defined later) recovery error and the one-step-ahead forecast error across our proposed method, the vanilla factor analysis using PC method a la \citet{stock2005implications}, and the proposed method in \citet{forni2005generalized}. 

\paragraph{An empirical algorithmic relaxation.} The actual implementation of Algorithm~\ref{algo:1} requires $\lambda_B$, $\lambda_\Theta$ as inputs, which in practice are challenging to select. On the other hand, the computation procedure designed for solving the convex program in~\eqref{eqn:optimization0_hetero} suggests that to obtain the estimates boils down to alternating between the following two steps: 
(1) a regularized regression (lasso) update on the rows of $B$; and (2) an SVT update on $\Theta$.  This naturally motivates the following steps in the implemented version of the algorithm, outlined next in Algorithm~\ref{algo:2}. 

% ========================================================
% IMPLEMENTATION ALGO 
\begin{algorithm}[!ht]
	\scriptsize
	\setstretch{0.5}
	\caption{\small An alternating minimization algorithm to obtain $\widehat{B}_{\text{emp}}$ and $\widehat{\Theta}_{\text{emp}}$.}\label{algo:2}
	%\SetKwInOut{Input}{\scshape input}
	%\SetKwInOut{Output}{\scshape output}
	\KwIn{Time series data $\{x_t\}_{t=0}^T$, tuning parameter $\lambda_B$, rank constraint $r$. }
	\BlankLine
	\textbf{Initialization:} Initialize with $\bar{\Theta}^{(0)}=\text{SVT}(\mathbf{X}_T)$\;
	\textbf{Iterate until convergence:} \\
	$\bullet$ \begin{minipage}[t]{15cm}
		Update $\bar{B}^{(m)}$ with the plug-in $\bar{\Theta}^{(m-1)}$ so that each row $j$ is obtained with Lasso regression (in parallel) and solves
		\begin{equation*}
		\bar{B}_{j\cdot} = \min_\beta \Big\{ \frac{1}{2T}\vertii{\big[\mathbf{X}_{T} - \bar{\Theta}^{(m-1)}\big]_{\cdot j} - \mathbf{X}_{T-1}\beta}^2  + \lambda_B \|\beta\|_1  \Big\}.
		\end{equation*}
	\end{minipage}\\
	$\bullet$  \begin{minipage}[t]{15cm}
		Update $\bar{\Theta}^{(m)}$ by singular value thresholding (SVT): do SVD on the lagged value-adjusted hyperplane, i.e., 
		\begin{equation*}
		\mathbf{X}_T - \mathbf{X}_{T-1}\bar{B}^{(m)} = UDV, \qquad \text{where}~~~D := \text{diag}\big(d_1,\dots,d_r,d_{r+1},\dots,d_{\min(T,p)}\big), 
		\end{equation*}
		and construct $\bar{\Theta}^{(m)}$ by 
		\begin{equation*}
		\bar{\Theta}^{(m)} = UD_rV, \qquad \text{where}~~~D_r:=\text{diag}\big(d_1,\dots,d_r,0,\dots,0\big).
		\end{equation*}
	\end{minipage}
	\BlankLine
	\KwOut{Estimated sparse transition matrix $\widehat{B}_{\text{emp}}=\bar{B}^{(\infty)}$ and the low rank hyperplane $\widehat{\Theta}_{\text{emp}}=\bar{\Theta}^{(\infty)}$.}
\end{algorithm}
% END OF ALGO
% =========================================================
\normalsize
Algorithm~\ref{algo:2} outlines the algorithmic relaxation to obtaining $(\widehat{B},\widehat{\Theta})$ in~\eqref{eqn:optimization0_hetero}, and it can be viewed as an alternating minimization algorithm that solves
\small
\begin{align}\label{eqn:optimization1}
(\widehat{\Theta}_{\text{emp}},\widehat{B}_{\text{emp}}):= \argmin~\big\{ \tfrac{1}{2T}\vertiii{\mathbf{X}_T - \Theta - \mathbf{X}_{T-1}B^\top}_{\F}^2  + \lambda_B \vertii{B}_1 \big\}, ~~~\text{subject to}~~\text{rank}(\Theta)\leq r.
\end{align}
\normalsize
For each update, the partial minimization step with respect to $\Theta$ or $B$ ensures that the value of the objective function is always non-ascending, which together with the fact that the objective function is bounded below guarantees convergence of the objective function iterates. In practice, the algorithm is terminated when the descent magnitude of the objective function between successive iterations is smaller than some pre-specified tolerance level.
This algorithm does not provide guarantees of convergence to a stationary point of the sequence of $(\bar{\Theta}^{(k)},\bar{B}^{(k)})$ iterates, which requires stronger assumptions --- either the convexity of the objective function and the constraint region, or the uniform compactness of the generated sequence of iterates. 

\paragraph{Choice of the tuning parameter $\lambda_B$ and the rank constraint $r$.} The implementation of Algorithm~\ref{algo:2} requires a specific pair of $(\lambda_B,r)$ as input. We consider choosing the optimal pair of $(\lambda_B,r)$ based on the information criterion proposed in \citet{ando2015selecting}, called the Panel Information Criterion (PIC) and defined as:
\begin{equation}\label{eqn:PIC}
\text{PIC}(\lambda_B,r) := \frac{1}{Tp}\smalliii{\mathbf{X}_T-\widehat{\Theta}_{\text{emp}} - \mathbf{X}_{T-1}\widehat{B}^\top_{\text{emp}}}_{\F}^2 + \widehat{\sigma}^2 \Big[ \frac{\log T}{T}\|\widehat{B}_{\text{emp}}\|_0 + r(\frac{T+p}{Tp})\log(Tp)\Big],
\end{equation}
where $\widehat{\sigma}^2 = \tfrac{1}{np}\smalliii{\mathbf{X}_n-\widehat{\Theta}_{\text{emp}} - \mathbf{X}_{n-1}\widehat{B}_{\text{emp}}^\top}_{\F}^2$ and $(\widehat{B}_{\text{emp}},\widehat{\Theta}_{\text{emp}})$ are solutions to~\eqref{eqn:optimization1} with the specific pair of plug-in $(\lambda_B,r)$. The optimal pair $(\lambda_B,r)$ is then selected in two steps: in step 1, we obtain $(\lambda_{B}^0,r^0)$ that gives the smallest PIC over a lattice $\mathcal{G}_{\lambda_B}\times \mathcal{G}_r := \{\lambda_B^{(1)},\dots,\lambda_B^{(j_1)}\}\times \{r^{(1)},\dots,r^{(j_2)} \}$; in step 2, we fix $r$ at $(d+1)\times r^0$ where $d$ is the number of lags corresponding to the sparse $\VAR(d)$ model, and seek for $\lambda_B^{\text{opt}}$ over a grid that minimizes PIC$(\lambda_B, (d+1)r^0)$. The optimal pair of tuning parameters is then given by $(\lambda_B^{\text{opt}},r^{\text{opt}}):=(\lambda_B^{\text{opt}},(d+1)\times r^0)$.

\paragraph{Data generating mechanism.} Synthetic data are generated according to the lag-adjusted factor model representation $X_t = \Lambda F_t + BX_{t-1} + \epsilon_t$. Starting from the standard approximate factor model representation $X_t = \widetilde{\Lambda} f_t + u_t$, $u_t$ is serially correlated and follows a $\text{sparse }\VAR(d)$ model\footnote{Throughout this section, we assume $d=1$; additional results for $d>1$ have been deferred to Supplement~\ref{appendix:VARd}.}, at each timestamp $t$, the $K$-dimensional factor is generated according to a $\VAR(q)$ model $f_t = \Phi_1 f_{t-1} +\cdots+\Phi_q f_{t-q} +  \eta_t$ where  $\eta_t\sim \mathcal{N}(0,\sigma_\eta^2\mathrm{I})$; decorrelating $u_t$ leads to following dynamic of $X_t$:
\begin{equation*}
X_t = \widetilde{\Lambda} f_t - B\widetilde{\Lambda} f_{t-1} + BX_{t-1} + \epsilon_t =: \Lambda F_t + BX_{t-1} + \epsilon_t, 
\end{equation*}
where $\Lambda = [\widetilde{\Lambda},B\widetilde{\Lambda}]$ and $F_t = (f_t^\top,f^\top_{t-1})^\top\in\R^{2K}$.

We consider several simulation settings as listed in Table~\ref{table:sim-setup} to test various facets of the model, primarily encompassing the dimensionality of the system $p$ and the number of factors $K$, as well as the sparsity structure of $B$ and its spectral radius that captures the level of autocorrelation. In addition to settings S0 to S4 where $\epsilon_t$ is Gaussian, to test the robustness of the proposed model to the presence of heavy tails, we consider also cases where $\epsilon_t$ follows some multivariate $t$ distribution (S5 to S7). Throughout all numerical experiments presented in this section, the sample size $T$ is fixed at 200 and the spectral radius of the $\VAR(q)$ system is set randomly from $\mathsf{Unif}[0.6,0.8]$. 
\begin{singlespace}
	\begin{table}[!h]
		\scriptsize
		\centering
		\captionsetup{font=scriptsize}
		\begin{tabular}{r|ccc|cc|l|c}
			\specialrule{1pt}{1pt}{1pt}
			& $p$ & sparsity structure of $B$ & $\varrho(B)$ & $K$ & $q$ & structure of $\Sigma_\epsilon$-dist. &   factor space/lag space relative strength  \\ \hline
			S0 & 100 & $2/p$, exactly sparse &  $0.7$ & 2 & 1 & diagonal - $\mathcal{N}$ & strong factor~~$\approx 3:2$ \\ 
			S1 & 100 & $5/p$, weakly sparse & 0.7 & 2 & 1 & Toeplitz(0.2) - $\mathcal{N}$ & strong factor~~$\approx 2:1$ \\
			S2 & 300 & $2/p$, weakly sparse & 0.7 & 5 & 1 & diagonal - $\mathcal{N}$ & strong factor~~$\approx 2:1$ \\
			S3 & 200 & $2/p$, exactly sparse & 0.9 & 5 & 2 & diagonal - $\mathcal{N}$ & strong lag~~$\approx 2:3$ \\
			S4 & 200 & $2/p$, weakly sparse & 0.7 & 5 & 4 & Toeplitz(0.2) - $\mathcal{N}$ & strong factor~~$\approx 3:2$\\ \hline
			S5 & 100 & $2/p$, exactly sparse & 0.7 & 5 & 1 & diagonal - $t_4$ & strong factor~~$\approx 3:2$\\
			S6 & 200 & $2/p$, weakly sparse & 0.7 & 5 & 1 & Toeplitz(0.2) - $t_8$ & strong factor~~$\approx 1:1$\\
			\specialrule{1pt}{1pt}{1pt}
		\end{tabular}
		\caption{Simulation settings for data generated according to a lag-adjusted dynamic factor model.}\label{table:sim-setup}
	\end{table}
\end{singlespace}
To generate the sparse transition matrix $B$, for each row that corresponds to the coefficients of each single time series regression, its (strong) support set is randomly generated to meet the specified density level (i.e., $2/p$ or $2/5$), and nonzero entries are then generated from $\pm\mathsf{Unif}[m_B-0.1,m_B+0.1]$. In the case of a weakly sparse $B$, entries in the weak support set are generated from $\mathsf{Unif}([-10\%m_B,10\%m_B])$. Finally, all entries are scaled to meet the specified $\varrho(B)$ level, to ensure that the system is stationary. For the dense factor loading matrix $\Lambda$, its entries are generated from $\pm\mathsf{Unif}[m_\Lambda-0.1,m_\Lambda+0.1]$. It is worth noting that the value of $m_B$ and $m_\Lambda$ are set so that the factor/lag space relative strength is satisfied, measured by the empirical relative signal-to-noise ratio for the $\Lambda F_t$ and the $BX_{t-1}$ component.

\paragraph{Performance evaluation.} To measure the accuracy of the obtained estimates and forecast, we focus on the following four components of the model:
\begin{itemize}
	\item[--] For the (weakly) sparse transition matrix $B$ we use sensitivity $\textrm{SEN} = \frac{\textrm{TP}}{\textrm{TP}+\textrm{FN}}$, specificity $\textrm{SPC} = \frac{\textrm{TN}}{\textrm{FP}+\textrm{TN}}$ and relative error in Frobenius norm ($\RERR_B$) as evaluation criteria. Note that in the case where $B$ is weakly sparse, despite the fact that entries in the weak support set are not exactly zero, they are effectively deemed as zeros for comparison purpose. 
	\item[--] For the factor hyperplane $\Theta$, since we don't separately identify the factors and in addition the factor space is invariant to identification restrictions, we measure its relative error in Frobenius norm ($\RERR_{\Theta}$), as well as its relative {\em projection error}, defined as $\text{ProjErr}_{\Theta}:= \smalliii{\Pi_{\widehat{\Theta}}-\Pi_{\Theta^\star}}_\F/\smalliii{\Pi_{\Theta^\star}}_\F$, where $\Pi_{\Theta^\star} := Q_{\Theta^\star}Q_{\Theta^\star}^\top$ with $Q_{\Theta^\star}$ being the orthonormal basis of $\Theta^\star$; $\Pi_{\widehat{\Theta}}$ can be analogously defined. Note that the following correspondence between $\sin\theta$ distance and the projection error holds: $\smalliii{\sin\theta(\widehat{\Theta},\Theta^\star)}_\F^2 = \tfrac{1}{2}\smalliii{\Pi_{\widehat{\Theta}}-\Pi_{\Theta^\star}}_\F$; moreover, this metric is not applicable in high-dimensional regimes ($p\geq T$) where it would stays at zero.
	\item[--] For the common space, in the case where it is estimated with the proposed lag-adjust DFM, at the population level it is captured by $BX_{t-1} + \Lambda F_t$ and hence its estimate is given by $\widehat{\Theta} + \mathbf{X}_{T-1}\widehat{B}^\top$; whereas in the case where the model is estimated based on the SW formulation or GDFM, the estimated factor space coincides with that of the common space. For all three models, we present the relative error in Frobenius norm of the estimates.
	\item[--] For the one-step-ahead forecast, we measure its squared $\ell_2$ norm w.r.t. the oracle $x_{T+1}^\star$, that is, $\|\widehat{x}_{T+1} - x^\star_T\|^2/\|x^\star_T\|^2$, where the oracle is given by $x^\star_{T+1} = B x_T + \Lambda F_{T+1}$ and can be viewed as the ``denoised'' version of $x_{T+1}$.
\end{itemize}

\begin{singlespace}
	\begin{table}[!h]
		\scriptsize
		\centering
		\captionsetup{font=scriptsize}
		\begin{tabular}{r|ccc|ccc|ccc|ccc}
			\specialrule{1pt}{1pt}{1pt}
			& \multicolumn{3}{c|}{$B$ recovery (lag-adj DFM)} & \multicolumn{3}{c|}{ $\Theta$ recovery (lag-adj DFM)} & \multicolumn{3}{c|}{common space recovery } & \multicolumn{3}{c}{one-step-ahead forecast}\\ 
			\cline{2-4} \cline{5-7} \cline{8-10} \cline{11-13}
			& SEN & SPC & $\RERR_B$ & $\widehat{K}$ & $\text{ProjErr}_{\Theta}$ & $\RERR_\Theta$ & lag-adj DFM & SW & GDFM & lag-adj DFM & SW & GDFM \\ \hline
			S0 &  0.99 & 0.98 & 0.28 & 2 & 0.15 & 0.20 & 0.13 & 0.32 & 0.31 & 0.51 & 0.60 & 0.53 \\
			S1 & 0.97 & 0.92 & 0.51 & 2 & 0.16 & 0.47 & 0.27 & 0.47 & 0.45 & 0.56 & 0.91 & 0.66 \\
			S2 & 0.99 & 0.95 & 0.74 & 5 & -- & 0.58 & 0.35 & 0.45 & 0.45 & 0.60 & 0.73 & 0.67\\
			S3 & 0.99 & 0.98 & 0.19 & 5 & -- & 0.26 & 0.22 & 0.72 & 0.50 & 0.36 & 0.92 & 0.90 \\
			S4 & 0.98 & 0.97 & 0.58 & 5 & -- & 0.51 & 0.32 & 0.44 & 0.44 & 0.47 & 0.58 & 0.57 \\ \hline
			S5 & 0.92 & 0.92 & 0.61 & 5 & 0.31 & 0.48 & 0.10 & 0.13 & 0.14 & 0.43 & 0.42 & 0.45 \\
			S6 & 0.98 & 0.93 & 0.47 & 5 & -- & 0.53 & 0.35 & 0.63 & 0.65 & 0.55 & 0.90 & 0.92 \\
			\specialrule{1pt}{1pt}{1pt}
		\end{tabular}
		\caption{Performance Evaluation for various Simulation Settings, median across 100 replications.}\label{table:sim-ours-results}
	\end{table}	
\end{singlespace}
As Table~\ref{table:sim-ours-results} demonstrates, for all three components, estimates obtained from Algorithm~\ref{algo:2} exhibit good performance. In particular, (i) the proposed method is robust to the sparsity structure of $B$, as both exactly-sparse and weakly-sparse settings yield very satisfactory strong support recovery (see S1, S2 and S4). (ii) A larger panel size $p$ leads to improved factor hyperplane recovery, as manifested in the form of smaller relative error in magnitude estimation although it requires the sparsity of the transition matrix to decrease accordingly (recall that it is set to $2/p$); however, the performance deteriorates as the dynamics of $f_t$ become more complex (e.g., S4). (iii) A strong signal in the lag-space leads to improved recovery of $B$, despite the presence of stronger temporal dependence which empirically incurs the algorithm to take more iterations to converge (e.g., S3). For all settings, PIC correctly selects the number of factors, which translates into the correct identification of the rank constraint. 

Next, we compare the performance of common space recovery and forecasting for the following three methods: the posited model, standard SW formulation and GDFM. For SW \citep{stock2005implications}, the reported error is based on the minimum error among estimates obtained under different rank constraints ranging between $K$ and $2K$; for GDFM \citep{forni2005generalized}, the reported error is based on the minimum error among estimates obtained under different combinations of $(\widetilde{q},\widetilde{r})$ that determines the number of common factors when loaded dynamically and contemporaneously. For all settings, the proposed method (lag-adjusted DFM) outperforms the other two by explicitly incorporating the lag space spanned by $X_{t-1}$; specifically, it outperforms its competitors by a wide margin when the lag space possesses a stronger signal, in which case SW becomes particularly susceptible (e.g., S3 and S6). However, as the dynamics of $f_t$ becomes more involved, its advantage becomes less pronounced (S4). 

Finally, the proposed model is relative robust to the presence of heavy tails, although the performance deteriorates compared to the Gaussian case. Specifically, when the distribution shows significant deviation from Gaussian (e.g., S5), the degradation manifests itself through less satisfactory recovery in the support of $B$ and larger error of the estimated factor space; whereas the forecasting performance isn't affected. On the other hand, with lighter tails (e.g., S6), the performance becomes comparable to the Gaussian case.

\subsection{Alternative DGPs.} 
To further compare the performance across all three methods, we consider settings where data generating processes deviate from the proposed model in~\eqref{model:ours}. Specifically, we adopt the data generating mechanism in \citet[Model II]{forni2017dynamic}, that is,
\begin{equation}
X_t  = \Lambda F_t + \xi_{t}~~\text{(Obsv Eqn)},~~~~F_t  = D F_{t-1} + Ku_t~~\text{(State Eqn)},
\end{equation}
where coordinates of $u_t$ and $\xi_t$ are i.i.d standard Gaussian white noises processes, with $u_t$ capturing the structural shocks. Entries of $\Lambda$ and $K$ are drawn independently from $\mathsf{Unif}[-1,1]$; entries of $D$ are first drawn independently from $\mathsf{Unif}[-1,1]$ then scaled so that the spectral norm of $D$ satisfies some pre-specified target, with the latter drawn from $\mathsf{Unif}[0.4,0.9]$. We focus on the performance of common space recovery and the one-step-ahead forecast, under various combinations of the model parameters, as listed in Table~\ref{table:sim-gdfm-results}:
\begin{singlespace}
	\begin{table}[!h]
		\scriptsize
		\centering
		\captionsetup{font=small}
		\begin{tabular}{ccc|ccc|ccc}
			\specialrule{1pt}{1pt}{1pt}
			\multirow{2}{*}{$\text{dim}(X_t)$} & \multirow{2}{*}{ $\text{dim}(u_t)$} & \multirow{2}{*}{ $\text{dim}(F_t)$} & \multicolumn{3}{c|}{common space recovery } & \multicolumn{3}{c}{one-step-ahead forecast}\\ 
			\cline{4-6} \cline{7-9} 
			&&&lag-adj DFM & SW & GDFM & lag-adj DFM & SW & GDFM \\  \hline
			100 & 2 & 4 &  0.212 (0.072) & 0.212 (0.072) & 0.208 (0.061) & 0.210 (0.812) & 0.210 (0.812) & 0.326 (4.528)\\
			100 & 4 & 4 & 0.144 (0.035) & 0.144 (0.035) & 0.143 (0.035) & 0.078 (0.324) & 0.078 (0.324) & 0.063 (0.768)\\
			200 & 4 & 6 & 0.119 (0.026) & 0.119 (0.026) & 0.121 (0.026) & 0.087 (0.302) & 0.087 (0.302) & 0.132 (1.679)\\ 
			300 & 6 & 6 & 0.100 (0.024) & 0.100 (0.024) & 0.085 (0.017) & 0.075 (0.351) & 0.075 (0.351) & 0.061 (0.263)\\
			\specialrule{1pt}{1pt}{1pt}
		\end{tabular}
		\caption{Simulation settings and performance evaluation, with DGP a la \citet{forni2017dynamic}}\label{table:sim-gdfm-results}
	\end{table}	
\end{singlespace}

As the results show, with the data generating procedure deviating from the proposed model in~\eqref{model:ours}, with properly chosen tuning parameters, the performance of the proposed methodology matches that of SW by effectively having $\widehat{B}=0$. Meanwhile, it is worth noting for all three methods, the performance of common space recovery shows significantly less variability compared with that of forecasting; in particular, the forecasting performance of GDFM exhibits the highest variance across replications, among the three methods.

%%%%%%%%%%% Real Data %%%%%%%%%%%%%%%%
\section{Application to Returns of US Financial Assets.}\label{sec:realdata}
Factor models have been widely used in financial applications. In particular, they have been employed in analyzing the dynamics of asset returns, either for the purpose of identifying risk factors, or for estimating the covariance structure amongst assets for better portfolio diversification and asset allocation \citep[e.g.,][]{fan2012vast}.  We applied the proposed modeling framework to a set of stocks return data corresponding to 75 large US financial institutions, which also exhibit strong (serial) correlation in the error terms. Specifically, we analyze the risk-free returns\footnote{The risk-free return of Stock $i$ at time $t$ is calculated as $\widetilde{r}_{i,t} = r_{i,t} - r_{\text{rf},t} = \frac{p_{i,t}-p_{i,(t-1)}}{p_{i,(t-1)}}- r_{\text{rf},t}$, where $p_{i,t}$ is its stock price at time $t$ and $r_{\text{rf},t}$ is the risk-free rate.} of 25 banks, 25 insurance companies and 25 broker/dealer firms for the period of 2001-17. Note that this time period contains a number of significant events for the financial industry, including the growth of mortgage bank securities \citep{ashcraft2010MBS} in the early 2000s, rapid changes in monetary policy in 2005-06, the great financial crisis \citep{eichengreen2010tale} in 2008-09 and the European debt crisis in 2011-12 and their aftermath. Our analysis identifies a number of interesting patterns, especially around the period 2007-09 encompassing the beginning, height and immediate aftermath of the US financial crisis, both through changes
in the factor structure and the partial autocorrelation one governed by the VAR model transition matrix of the log-returns of these financial assets.

\paragraph{Data.} The data consist of weekly stock return data corresponding to 75 large financial institutions in terms of market capitalization, for the period of January 2001 to December 2017 and were obtained from the Center for Research in Security Prices (CRSP) database. The 75 companies are categorized into three sectors: banks (SIC code 6000--6199), broker/dealers (SIC code 6200--6299) and insurance companies (SIC code 6300--6499), with 25 in each sector \citep[see also][]{billio2012econometric}. As we require that the data be available for the entire time span under consideration, 56 firms are kept for further analysis, since the remaining ones either went bankrupt or were forced to merge with financially healthier companies (e.g. Lehman Brothers and Merill Lynch in 2008, resp.). To get an overview of the correlation structure amongst the stocks after accounting for the first principal component that
captures the weighted average return of the portfolio they constitute \citep{avellaneda2010statistical}, we plot the correlation among the principal component regression residuals. Specifically, the entire time span is broken into three sub-periods that have been previously considered in the literature \citep[c.f.][]{billio2012econometric}: 2001--2006 (pre-crisis), 2007--2009 (crisis), 2010-2017 (post-crisis), and plot the correlation maps corresponding to samples in each period. As Figure~\ref{fig:corr} demonstrates, overall, we observe positive correlation within each sector and negative correlation across them. Such a structural pattern is predominant in the pre-crisis period especially within the insurance sector, and becomes significantly weaker in the post crisis one; whereas during the crisis, stronger negative correlation across blocks is present as well as scattered positive correlations. This suggests that different factor and auto-regressive structures emerge during the crisis period. Further, note that similar results hold if we examine the residuals after removing a second principal component, so as to capture a larger percentage of variance of the stock returns.
\begin{figure}[h]
	\captionsetup{font=scriptsize}
	\centering
	\begin{boxedminipage}{2in}
		\includegraphics[scale=0.4]{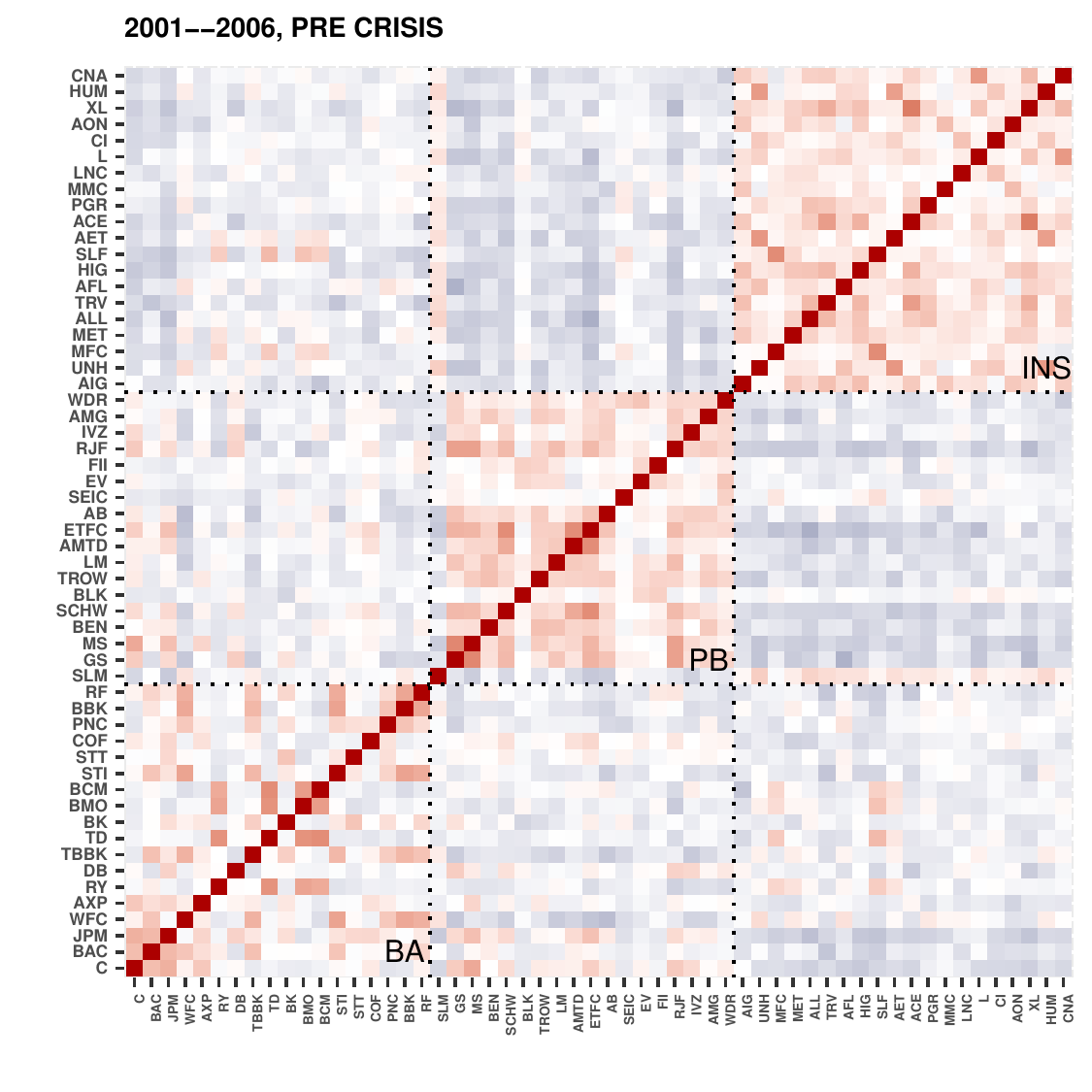}
	\end{boxedminipage}
	\begin{boxedminipage}{2in}
		\includegraphics[scale=0.4]{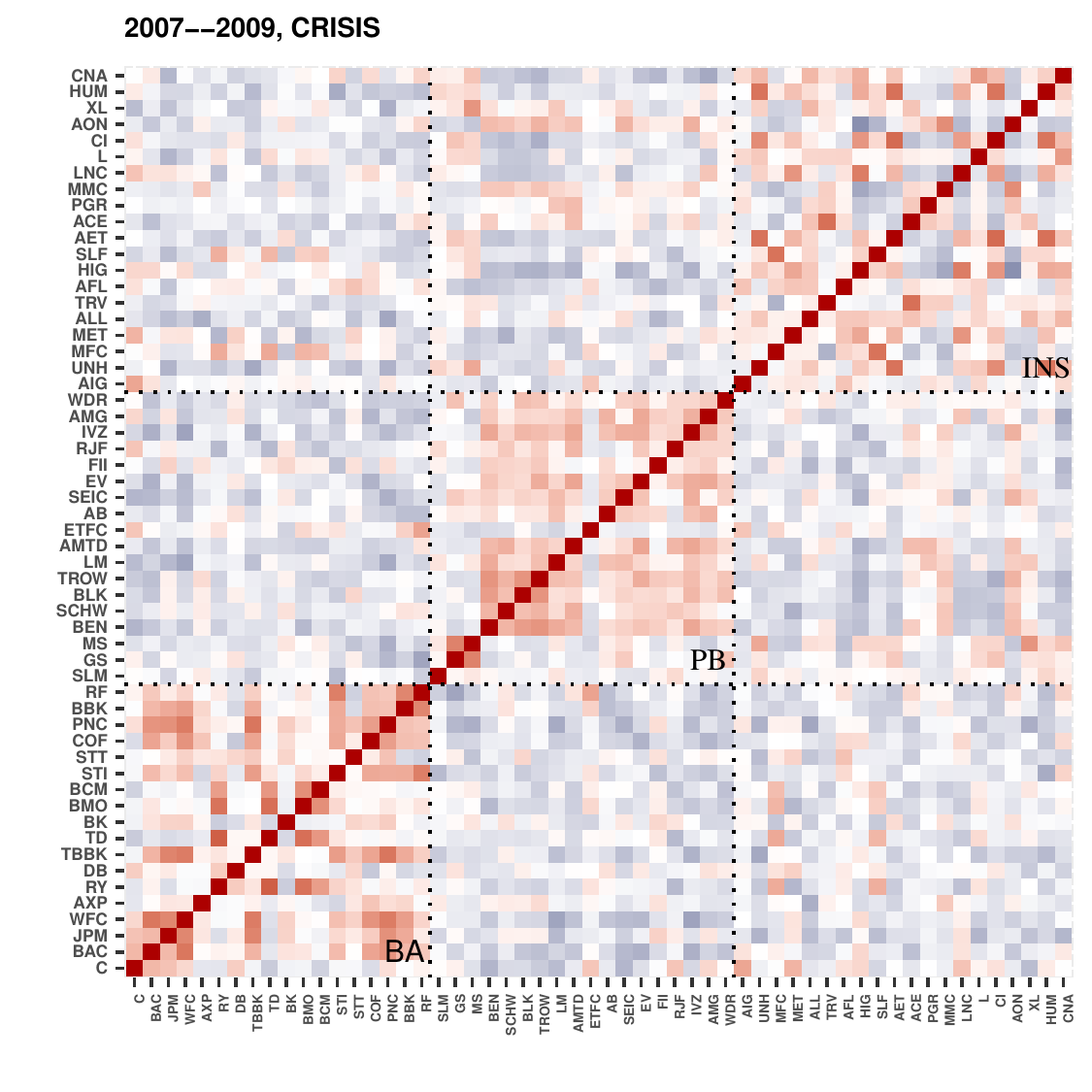}
	\end{boxedminipage}
	\begin{boxedminipage}{2in}
		\includegraphics[scale=0.4]{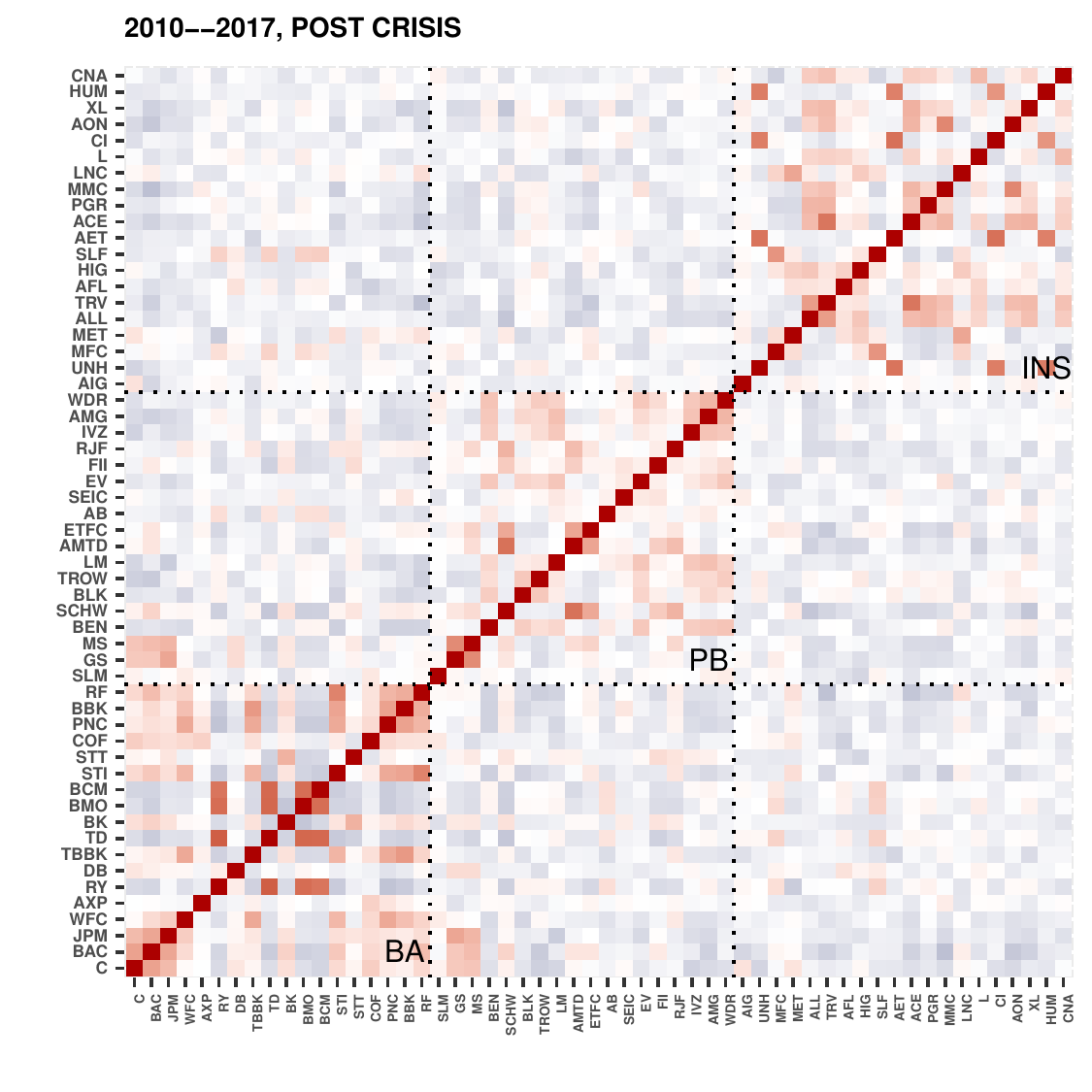}
	\end{boxedminipage}
	\caption{Correlation map for PCR residuals with the number of PC fixed at 1. From left to right: 2001--2006, pre-crisis; 2007-2009, crisis;  2010--2017, post-crisis. Red to blue corresponds to correlation from 1 to -1.}\label{fig:corr}
\end{figure}

The analysis is based on 104-week-long rolling windows to avoid issues with non-stationarity that potentially depends on length of the period under consideration. This strategy has also been used in \citep{billio2012econometric, lin2017regularized} and allows monitoring change in the number of factors over time, as well as the sparsity level of $B$ which measures the connectivity of the partial autocorrelation network across these financial institutions. Note that 48\% of the rolling samples fail to reject the null hypothesis that they are multivariate normality\footnote{we consider the Henze-Zirkler's multivariate normality test.}, with the exception of samples during the crisis period and those at the end of the sampling period post-crisis (around 2015--17). We fit the proposed lag-adjusted factor model in each time window, with tuning parameters selected according to a modified PIC criterion\footnote{The criterion is modified to $\text{PIC}^*(B,r) = \log\widehat{\sigma}^2 + \big(\frac{\log T}{T}\big)\|B\|_0 + r\cdot \big(\frac{T+p}{Tp}\big)\log(Tp)$.} that does not depend on the range within which the number of factors is being searched. 
\begin{figure}[h]
	\captionsetup{font=scriptsize}
	\centering
	\includegraphics[scale=0.7]{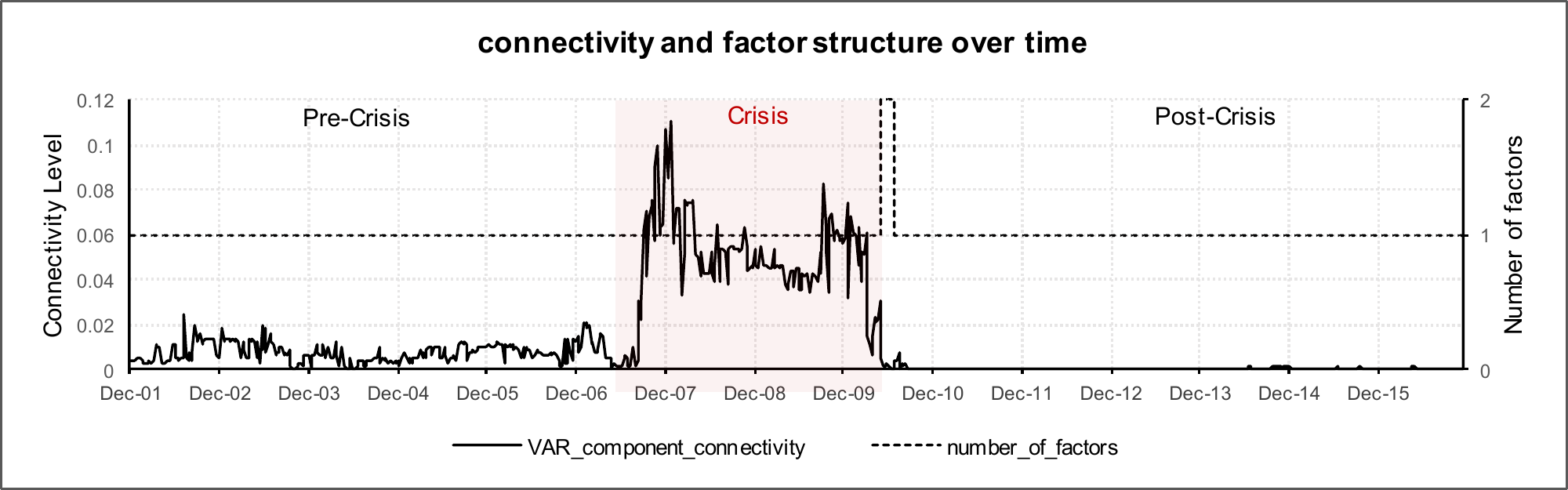}
	\caption{Results after fitting the model to the real data based on 104-week-long rolling windows over time. Left panel: number of factors (right axis, in red) and the connectivity level of $\widehat{B}$ (left axis, in blue). Right panel: overall R-squared (solid line) and the R-squared attributed to the factor (dotted line). Time stamps on the horizontal axis correspond to the mid-point of the window in question.}\label{fig:realdata1}
\end{figure}
As Figure~\ref{fig:realdata1} shows, sharp changes are observed in the temporal dependence structure of stock returns during the crisis period. In particular, two change points respectively correspond to the beginning of the 2007 sub-prime mortgage crisis and the ending of the 2008--2009 global financial crisis. Specifically, for the pre- and post-crisis periods, the density of the transition matrix stays at a level close to zero, suggesting that not much serial correlation exists in the idiosyncratic component after the common factor (proxy for the market portfolio) is accounted for. During the crisis period, however, the connectivity level of $\widehat{B}$ witnesses a sharp increase, reaching its maximum in the sampling window corresponding to Dec 2006--Dec 2008, during which period multiple major events of the financial crisis occurred. We also track the change in R-squared and the R-squared attributed to the factor over time, as a surrogate for the quality of the model fit, as shown in Table~\ref{table:Rsq}.
\begin{table}[H]
	\centering\small
	\captionsetup{font=scriptsize}
	\begin{tabular}{l|ccc}
		\hline
		& pre-crisis & crisis & post-crisis \\ \hline
		Total Rsq & 42.11\% & 58.18\% & 56.76\% \\
		Factor Rsq & 41.29\% & 	54.39\% & 56.50\% \\ \hline
	\end{tabular}
	\caption{R-squared averaged across all stock for each sub-periods.}\label{table:Rsq}
\end{table}
In accordance with the connectivity level of the VAR component that captures the temporal dependence among the stocks, under normal market conditions, the majority fit of the model (R-squared) comes from the factor hyperplane; whereas during the crisis period, the gap between the total R-squared and the factor R-squared widens with the lag term explaining a non-negligible proportion of the R-squared, which indicates the presence of significant cross autocorrelations in the returns. 

To further investigate the factor composition and the temporal dependence structure during the crisis period, we zoom in on the sampling frequency and focus on the year of 2008. Specifically, we consider daily data from January 2008 to December 2008 that cover 253 consecutive trading days and fit the proposed lag-adjusted factor model. Note that for this part of the analysis, the sample consists of 72 stocks. Using $\text{PIC}^*$, 2 factors are identified, with the dominant one capturing 55\% of the R-squared followed by 11\% from the second. For reconstruction purposes, we assume they are orthogonal so that the factor composition can be retrieved from the singular vectors of $\widehat{\Theta}$. 

As depicted in the left panel of Figure~\ref{fig:realdata3}, all financial institutions contribute positively to the first factor, with dominating contributors spread in all sectors. The composition of the second factor shows an interesting pattern: two negative contributors are FRE (Freddie Mac) and FNM (Fannie Mae), and the positive ones are primarily in the insurance sector. However, AIG---unlike its peers---shows almost zero contribution to the second factor, albeit its strong contribution to the first one. The latter is consistent with other findings that it played a prominent role during the crisis.\footnote{According to an estimate as of January 2010, AIG accounted for 38\% of the total losses incurred by insurance companies (98.2 out of 261.0 billions) since 2007. Source: Bloomberg, see also \citet{schich2010insurance}.} 
In the right panel of Figure~\ref{fig:realdata3}, we plot the partial auto-correlation network of the firms during the crisis after properly thresholding the entries that have small magnitudes, with red edges denoting positive links and with blue negative ones. Nodes that belong to the same sector are colored identically. A careful examination of the node weighted in/out-degrees shows that the top emitters are relatively uniform, in the sense that their weighted out-degrees do not differ by much; whereas the top receivers are dominant, since the weighted in-degrees for top receivers are significantly higher compared with the rest. Further top emitters heavily concentrate in the insurance sector. Meanwhile, some of the top receivers are also major contributors to the factors' composition, e.g., AIG to the 1st factor, HIG to the 2nd, etc. This finding partially aligns with the role that many insurance companies played in magnifying the impact of the crisis on the overall stability of the financial system, due to their large insurance underwriting of Credit Default Swaps and subsequent exposure to accentuated risks \citep{eichengreen2010tale}. However, this analysis points to the importance of insurance companies based on publicly available data and before their role in the crisis was fully revealed and understood. It is worth noting that with the same set of data, vanilla factor analysis using the information criterion proposed in \citet{bai2002determining} only identifies 1 factor, which further substantiates the aforementioned point that classical factor analysis may lead to skewed inference when strong correlation among the idiosyncratic component is present.
\begin{figure}[h]
	\centering
	\captionsetup{font=scriptsize}
	\begin{boxedminipage}{1in}
		\includegraphics[scale=0.715]{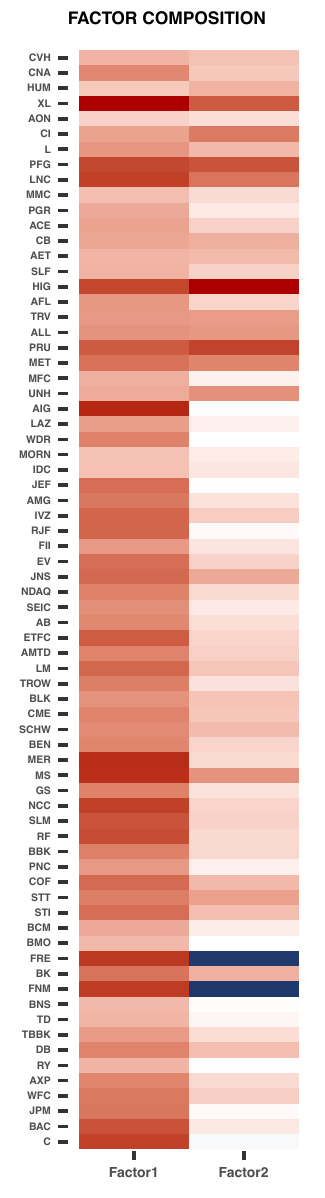}
	\end{boxedminipage}
	\begin{boxedminipage}{3.8in}
		\includegraphics[scale=0.444]{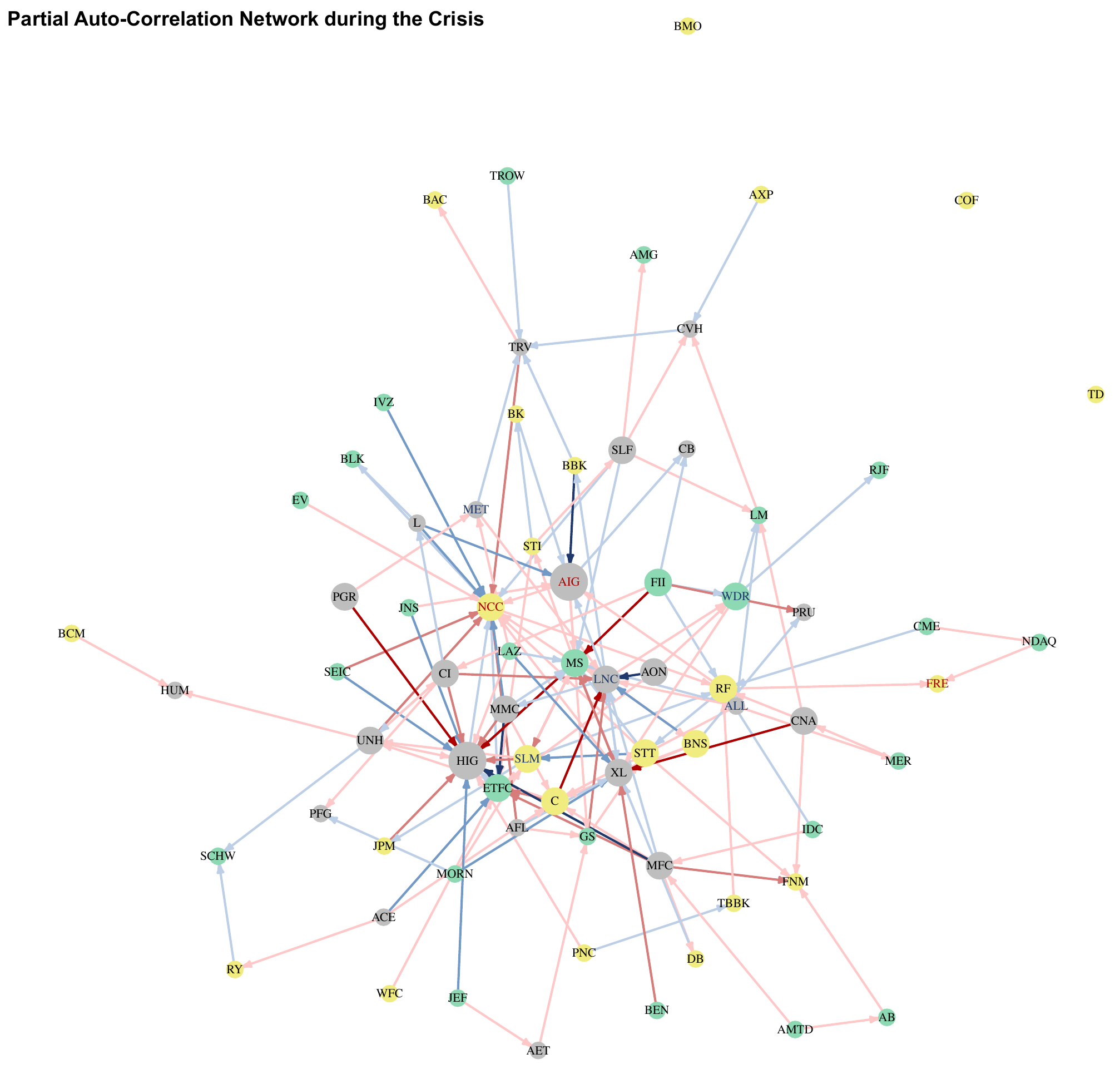}
	\end{boxedminipage}
	\caption{Left panel: factor composition. Right panel: partial autocorrelation network during the crisis, after proper thresholding of entries with small magnitudes. Top 5 emitters: MFC, MMC, MS, AON, PGR. Top receivers (in white bars): HIG, NCC, LNC, XL, AIG. Node colors indicate their respective sector (gray--INS, yellow-BA, green--PB).}\label{fig:realdata3}
\end{figure}

To conclude this section, we compare and contrast our results with those obtained in \citet{billio2012econometric}, in which the authors consider 100 financial institutions comprising of the largest 25 among each of the four categories: hedge funds, broker/dealers, banks and insurers; thus, that data set is enhanced by the inclusion of big hedge funds for which publicly available stock quotes are not accessible. From the systemic risk standpoint, the authors measure the connectedness of the system based on principal component analysis (PCA) and Granger-causality network analysis during the 1994--2008 period, and  identify increased level of interconnectedness during the crisis period and the asymmetry in the degree of connectedness amongst different sectors. Our results are qualitatively similar to these results, and the conclusions broadly match. However, we would like to highlight some key differences in both modeling and in the empirical results obtained. From the modeling perspective, \citet{billio2012econometric} consider two separate modeling 
strategies: (i) a Principal Components Analysis (akin to a static factor model) and (ii) a Granger-causality based analysis through fitting a VAR model for each pair of stocks returns. The PCA analysis examines a {\em fixed} number of principal components/factors and the authors argue that the increasing proportion of variation explained by them is an indication of the systematic response of the financial system to the crisis. Their pairwise based Granger-causal network also reveals increased connectivity during the crisis period. Our model considers latent factors and lead-lag relationships among stock returns {\em simultaneously}, thus gaining better and more informative insights. In addition, the lead-lag relationships are considered across all firms simultaneously rather than in a pairwise fashion. By incorporating the strong correlations present in the idiosyncratic component, our model is more parsimonious. Specifically, during the crisis, \citet{billio2012econometric} uses 10 principal components to account for 85\% of the returns variance, whereas only 5 suffice in our model; further, the leading factor in their analysis only accounts 37\% of the variance, compared to 50\% in our model.
Finally, extending the analysis period to 2017 shows that after 2011 the influence of banks and insurance companies on stock returns waned, as the marker slowly returned to normalcy. However, we are in broad agreement with the \citet{billio2012econometric} conclusion on the heightened role of banks and insurers up to 2009.

%%%%%%%%%%%%%%%%%%%%%%%%%%%%%% Discussion %%%%%%%%%%%%%%%%%%%%%%%%%%%%%%%%%%
\section{Discussion.}\label{sec:discussion}

In this paper, we introduced a novel modeling framework that generalizes the classical approximate factor model to include lags of the observable process, so that stronger correlations among the idiosyncratic component can be accommodated. The autoregressive structure is assumed to be sparse, which enables its estimation for large time series panels. Estimation of the model parameters is based on a maximum likelihood formulation that leads to a convex optimization problem, and the resulting estimates come with high probability error bounds guarantees that can be expressed in terms of key structural parameters ($n, p, K, s$, etc.), and exhibit superior empirical performance in synthetic data.

In addition to generalizing the model in \citet{stock2005implications}, our proposed model can also be perceived as a robust treatment of endogeneity. Specifically, as noted by \citet{anderson2007forecasting}, in the presence of large values in $u_t$ and for a relatively small panel size $p$, the factor estimates will be distorted as a result of this endogeneity. In this work, by explicitly taking into consideration the lagged terms in the dynamics of $X_t$, the noise term $\epsilon_t$ becomes strictly exogenous. Our proposed model and estimation procedure has the capacity of handling much stronger correlation between $F_t$ and $u_t$, although ultimately we do require $\Cov(F_t,u_t)$ to be indirectly bounded in some appropriate way.

%\vskip 1in
\clearpage
\begin{singlespace}
\bibliography{refbib}
\end{singlespace}

\clearpage
\setcounter{page}{1}
\pagenumbering{arabic}

\titleformat{\section}{\large\normalfont\centering\bfseries}{APPENDIX~\thesection.}{0.5em}{\uppercase}
\renewcommand{\thesubsection}{\Alph{section}.\arabic{subsection}}
\numberwithin{lemma}{section}
\numberwithin{remark}{section}
\appendix

%%%%%%%%%%%% Proof of Theorem %%%%%%%%%%%%%%%%%
\section{Proofs for Statistical Error Bounds. }\label{appendix:proof-of-theory-convex}
%%%%%%%%%%%%%%%%%%%%%%%%%%%%%%%%%%%%%%%%%%%%%%%%%%%%%%
%***********************************************
% PROOF OF THEOREM 1
%***********************************************
%%%%%%%%%%%%%%%%%%%%%%%%%%%%%%%%%%%%%%%%%%%%%%%%%%%%%%
Before presenting the proof of Theorem~\ref{thm:fix-bound-convex}, we first define a few quantities associated with the regularizers. Given some $\eta$ (to be specified later), let $S^\star_\eta$ denote the thresholded support set of $B^\star$, i.e.,
$S^{\star}_\eta := \{(i,j):|B^\star_{ij}|>\eta\}$, and let the SVD of $\Theta^\star$ be $\Theta^\star = (U^\star)D^\star(V^\star)^\top$, with $U^\star_K$ and $V^\star_K$ respectively denoting the first $K$ columns of $U^\star$ and $V^\star$. Let $\mathbb{S}$, $\mathbb{M}$ and their complements respectively be defined as follows: 
\begin{align*}
\begin{split}
\mathbb{S}:= \big\{ \Delta\in\R^{p\times p}~|~\Delta_{ij}=0\text{ for }(i,j)\notin S^\star_\eta \big\}, \\
\mathbb{S}^c:= \big\{ \Delta\in\R^{p\times p}~|~\Delta_{ij}=0\text{ for }(i,j)\in S^\star_\eta \big\} ,
\end{split}
\end{align*}
and
\begin{align*}
\begin{split}
\mathbb{M}:= \big\{ \Delta\in\R^{T\times p}~|~\text{row}(\Delta)\subseteq V^\star_K\text{ and }\text{col}(\Delta)\subseteq U^\star_K \big\}, \\
\mathbb{M}^\perp:= \big\{ \Delta\in\R^{T\times p}~|~\text{row}(\Delta)\perp V^\star_K\text{ and }\text{col}(\Delta)\perp U^\star_K \big\}.
\end{split}
\end{align*}
Further, for some generic matrix $\Delta_1\in\R^{p\times p}$, we define its projection on $\mathbb{S}$ and $\mathbb{S}^c$ (denoted by $\Delta_{1|\mathbb{S}}$ and $\Delta_{1|\mathbb{S}^c}$, resp.) as
\begin{equation}\label{eqn:projS}
\Delta_{1|\mathbb{S},ij}:= \mathbf{1}\big\{(i,j)\in S^\star\}\Delta_{1,ij} \qquad \text{and} \qquad 
\Delta_{1|\mathbb{S}^c,ij}:= \mathbf{1}\big\{(i,j)\notin S^\star\}\Delta_{1,ij}.
\end{equation}
With the above definitions and projections, we can write
\begin{equation}\label{eqn:Delta1}
\Delta_1 = \Delta_{1|\mathbb{S}} +  \Delta_{1|\mathbb{S}^c}, \qquad \|\Delta_1\|_1 = \|\Delta_{1|\mathbb{S}}\|_1  + \|\Delta_{1|\mathbb{S}^c}\|_1, \qquad \forall~\Delta_1\in\R^{p\times p},
\end{equation} 
and note that the following inequality holds:
\begin{equation}\label{eqn:BSnorms}
\|\Delta_{1|\mathbb{S}}\|_1\leq \sqrt{s}\smalliii{\Delta_{1|\mathbb{S}}}_{\F}\leq \sqrt{s}\smalliii{\Delta_{1}}_{\F}. 
\end{equation}
as $\Delta_{1|S}$ has at most $s$ nonzero entries where $s:=|S^\star|$. In an analogous way, for some generic matrix $\Delta_2\in\R^{T\times p}$, its projections on $\mathbb{M}$ and $\mathbb{M}^\perp$ (denoted by $\Delta_{2|\mathbb{M}}$ and $\Delta_{2|\mathbb{M}^\perp}$, resp.) are defined as
\begin{equation}\label{eqn:projM}
\Delta_{2|\mathbb{M}} := U^\star \begin{bmatrix}
\widetilde{\Delta}_{2,11} & \widetilde{\Delta}_{2,12} \\ \widetilde{\Delta}_{2,21} & O
\end{bmatrix}(V^\star)^\top \qquad \text{and} \qquad \Delta_{2|\mathbb{M}^\perp} := U^\star \begin{bmatrix}
O & O \\ O & \widetilde{\Delta}_{2,22}
\end{bmatrix}(V^\star)^\top,
\end{equation}
where $\widetilde{\Delta}_2$ is defined below and partitioned as:
\begin{equation*}
\widetilde{\Delta}_2 = (U^\star)^\top \Delta_2 (V^\star) = \begin{bmatrix}
\widetilde{\Delta}_{2,11} & \widetilde{\Delta}_{2,12} \\ \widetilde{\Delta}_{2,21} & \widetilde{\Delta}_{2,22}
\end{bmatrix}, \qquad \text{with}~~\widetilde{\Delta}_{2,11}\in\R^{K\times K}.
\end{equation*}
Note that the following relationships hold
\begin{equation}\label{eqn:Delta2}
\Delta_2 = \Delta_{2|\mathbb{M}} + \Delta_{2|\mathbb{M}^\perp}, \qquad \smalliii{\Delta_2}_* = \smalliii{\Delta_{2|\mathbb{M}} + \Delta_{2|\mathbb{M}^\perp}}_* = \smalliii{\Delta_{2|\mathbb{M}}}_* + \smalliii{\Delta_{2|\mathbb{M}^\perp}}_*,\quad  \forall~\Delta_2\in\R^{T\times p}.
\end{equation}
Next, we introduce concepts and lemmas regarding {\em decomposable regularizers} \citep{negahban2012unified}. Define the weighted regularizer as
\begin{equation*}
\mathcal{R}(B,\Theta):= \|B\|_1 + \frac{\lambda_\Theta}{\lambda_B} \smalliii{\Theta/\sqrt{n}}_*,
\end{equation*}
and let $\Delta_B:=\widehat{B}-B^\star$ and $\Delta_\Theta:=\widehat{\Theta}-\Theta^\star$.

\begin{lemma}\label{auxLemma:triangle} 
	With the definitions of projections in~\eqref{eqn:projS} and~\eqref{eqn:projM}, the following inequality holds:
	\begin{equation*}
	\mathcal{R}(B^\star,\Theta^\star) - \mathcal{R}(\widehat{B},\widehat{\Theta}) \leq \mathcal{R}(\Delta_{B|\mathbb{S}},\Delta_{\Theta|\mathbb{M}}) - \mathcal{R}(\Delta_{B|\mathbb{S}^c},\Delta_{\Theta|\mathbb{M}^\perp}) + 2\mathcal{R}(B^\star_{\mathbb{S}^c}, \Theta^\star_{\mathbb{M}^\perp}).
	\end{equation*}
\end{lemma}

\begin{lemma}\label{auxLemma:rank} 
	With the definition of~\eqref{eqn:projM}, the following holds for some generic $\Delta\in\R^{T\times p}$:
	$$\rank(\Delta_{\mathbb{M}})\leq 2\cdot \rank(\Theta^\star).$$
\end{lemma}

The proofs of these two lemmas are deferred to Supplement~\ref{appendix:auxLemma}. Based on the above preparatory steps, we present next the proof of Theorem~\ref{thm:fix-bound-convex}.
\begin{proof}[Proof of Theorem~\ref{thm:fix-bound-convex}]
	We prove the bound for $\Delta_B:=\widehat{B}-B^\star$ and $\Delta_\Theta:=\widehat{\Theta}-\Theta^\star$ under the imposed regularity conditions, where  $(\widehat{B},\widehat{\Theta})$ is the solution to the optimization problem~\eqref{eqn:optimization0}. Using the optimality of $(\widehat{B},\widehat{\Theta})$ and the feasibility of $(B^\star,\Theta^\star)$, the following {\em basic inequality} holds on:
	\begin{equation}\label{eqn:basic1}
	\begin{split}
	\frac{1}{2T} \smalliii{\mathbf{X}_{T-1}\Delta_B^\top + \Delta_{\Theta}}_{\F}^2 \leq & \frac{1}{T}\Big(\llangle \Delta_B^\top,\mathbf{X}_{T-1}^\top \mathbf{E}\rrangle + \llangle \Delta_{\Theta},\mathbf{E}\rrangle  \Big)  \\
	& + \lambda_B\Big(\smallii{B^\star}_1 - \smallii{\widehat{B}}_1 \Big) + \lambda_\Theta\Big(\smalliii{\Theta^\star/\sqrt{T}}_* - \smalliii{\widehat{\Theta}/\sqrt{T}}_*\Big). 
	\end{split}
	\end{equation}
	The LHS can be equivalently written as
	\begin{equation*}
	\frac{1}{2T}\smalliii{\mathbf{X}_{T-1}\Delta_B^\top + \Delta_\Theta}_{\F}^2 = \frac{1}{2T}\Big(\smalliii{\mathbf{X}_{T-1}\Delta_B^\top}_{\F}^2 + \smalliii{\Delta_\Theta}_{\F}^2 + 2 \llangle \mathbf{X}_{T-1}\Delta_B^\top,\widehat{\Theta}-\Theta^\star\rrangle \Big),
	\end{equation*}
	and by rearranging,~\eqref{eqn:basic1} becomes
	\begin{align}\label{eqn:basic2}
	\begin{split}
	\frac{1}{2T}\smalliii{\mathbf{X}_{n-1}\Delta_B^\top}_{\F}^2 + \frac{1}{2}\smalliii{\Delta_\Theta/\sqrt{T}}_{\F}^2   \leq \frac{1}{T}\llangle \mathbf{X}_{n-1}\Delta_B^\top,\widehat{\Theta}-\Theta^\star \rrangle + \frac{1}{T} \Big( \llangle \Delta_B^\top,\mathbf{X}_{n-1}^\top\mathbf{E}\rrangle + \llangle \Delta_{\Theta},\mathbf{E} \rrangle\Big)   \\
	+ \lambda_B\Big(\smallii{B^\star}_1 - \smallii{\widehat{B}}_1 \Big) + \lambda_\Theta\Big(\smalliii{\Theta^\star/\sqrt{T}}_* - \smalliii{\widehat{\Theta}/\sqrt{T}}_*\Big). 
	\end{split}
	\end{align}
	Based on~\eqref{eqn:basic2}, the rest of the proof is divided into three parts: in part (i), we provide a lower bound for the LHS primarily using the RSC condition; in part (ii), we provide an upper bound for the RHS with the designated choice of $\lambda_B$ and $\lambda_\Theta$; in part (iii), we align the two sides and obtain the error bound after some rearrangement. 
	\paragraph{Part (i).}  In this part, we obtain a lower bound for the LHS of~\eqref{eqn:basic2}. Using the RSC condition for $\mathbf{X}_{n-1}$, the following lower bound holds for the LHS of~\eqref{eqn:basic2}:
	\begin{equation}\label{eqn:LHSlow1}
	\frac{1}{2T}\smalliii{\mathbf{X}_{T-1}\Delta_B^\top}_{\F}^2 + \frac{1}{2}\smalliii{\Delta_\Theta/\sqrt{T}}_{\F}^2  \geq \frac{\alpha_{\text{RSC}}}{2}\smalliii{\Delta_B}_{\F}^2 + \frac{1}{2}\smalliii{\Delta_{\Theta}/\sqrt{T}}_{\F}^2 - \tau_T\|\Delta_B\|_1^2.
	\end{equation}
	To further lower-bound~\eqref{eqn:LHSlow1}, consider an upper bound for $\smallii{\Delta_B}_1$ with the aid of~\eqref{eqn:basic1}. By H\"{o}lder's inequality, the following inequalities hold for the inner products:
	\begin{equation}\label{eqn:holder1}
	\frac{1}{T}\llangle \Delta_B^\top,\mathbf{X}_{T-1}^\top\mathbf{E} \rrangle \leq \|\Delta_B\|_1\|\mathbf{X}_{T-1}^\top \mathbf{E}/T\|_\infty, 
	\end{equation}
	and
	\begin{equation}\label{eqn:holder2}
	\frac{1}{T}\llangle \Delta_\Theta,\mathbf{E} \rrangle \leq \smalliii{\Delta_\Theta/\sqrt{T}}_*\smalliii{\mathbf{E}/\sqrt{T}}_{op} = \smalliii{\Delta_{\Theta}/\sqrt{T}}_*\Lambda^{1/2}_{\max}(S_{\mathbf{E}}).
	\end{equation}
	By choosing $\lambda_B\geq 2\|\mathbf{X}_{T-1}^\top\mathbf{E}/T\|_\infty$ and $\lambda_\Theta\geq ~\Lambda^{1/2}_{\max}(S_{\mathbf{E}})$, the following inequality can be derived from the non-negativity of the RHS in~\eqref{eqn:basic1}:
	{\small
		\begin{equation*}
		\begin{split}
		0 &\leq \frac{\lambda_B}{2}\|\Delta_B\|_1 + \lambda_\Theta\smalliii{\Delta_{\Theta}/\sqrt{T}}_* + \lambda_B\mathcal{R}(B^\star,\Theta^\star) - \lambda_B\mathcal{R}(\widehat{B},\widehat{\Theta}) \\ 
		& \stackrel{(1)}{\leq} \frac{\lambda_B}{2}\|\Delta_{B|\mathbb{S}}\|_1 + \frac{\lambda_B}{2}\|\Delta_{B|\mathbb{S}^c}\|_1 + \lambda_\Theta \smalliii{\tfrac{\Delta_{\Theta|\mathbb{M}}}{\sqrt{T}}}_* + \lambda_\Theta \smalliii{\tfrac{\Delta_{\Theta|\mathbb{M}^\perp}}{\sqrt{T}}}_* \\
		&\qquad + \lambda_B\Big( \mathcal{R}(\Delta_{B|\mathbb{S}},\Delta_{\Theta|\mathbb{M}}) - \mathcal{R}(\Delta_{B|\mathbb{S}^c},\Delta_{\Theta|\mathbb{M}^\perp}) + 2\mathcal{R}(B^\star_{\mathbb{S}^c}, \Theta^\star_{\mathbb{M}^\perp})  \Big),
		\end{split}
		\end{equation*}
	}%
	\normalsize
	where the first two terms in (1) come from~\eqref{eqn:Delta1}, the next two terms come from \eqref{eqn:Delta2} and the last three terms use Lemma~\ref{auxLemma:triangle}. After writing out $\mathcal{R}(\cdot,\cdot)$ and rearranging, we obtain
	\begin{equation*}
	\begin{split}
	\frac{\lambda_B}{2} \|\Delta_{B|\mathbb{S}^c}\|_1 \leq \frac{3\lambda_B}{2} \|\Delta_{B|\mathbb{S}}\|_1 + 2\lambda_\Theta \smalliii{\tfrac{\Delta_{\Theta|\mathbb{M}}}{\sqrt{T}}}_* + 2\mathcal{R}(B^\star_{|\mathbb{S}^c}, \Theta^\star_{|\mathbb{M}^\perp}) , \\ 
	\frac{\lambda_B}{2} \|\Delta_{B|\mathbb{S}^c}\|_1  + \frac{\lambda_B}{2} \|\Delta_{B|\mathbb{S}}\|_1  \leq 2\lambda_B\|\Delta_{B|\mathbb{S}}\|_1 + 2\lambda_\Theta \smalliii{\tfrac{\Delta_{\Theta|\mathbb{M}}}{\sqrt{T}}}_* + 2\mathcal{R}(B^\star_{\mathbb{S}^c}, \Theta^\star_{\mathbb{M}^\perp}),
	\end{split}
	\end{equation*}
	that is, 
	\begin{equation}\label{eqn:B1up}
	\|\Delta_B\|_1 \leq 4\mathcal{R}(\Delta_{B|\mathbb{S}},\Delta_{\Theta|\mathbb{M}}) + 4\mathcal{R}(B^\star_{|\mathbb{S}^c}, \Theta^\star_{|\mathbb{M}^\perp}).
	\end{equation}
	Note that for $\mathcal{R}(\Delta_{B|\mathbb{S}},\Delta_{\Theta|\mathbb{M}})$, using~\eqref{eqn:BSnorms} and Lemma~\ref{auxLemma:rank},
	\begin{equation}\label{eqn:Rup}
	\begin{split}
	\mathcal{R}(\Delta_{B|\mathbb{S}},\Delta_{\Theta|\mathbb{M}}) & = \|\Delta_{B|\mathbb{S}}\|_1 + \frac{\lambda_\Theta}{\lambda_B} \smalliii{\tfrac{\Delta_{\Theta|\mathbb{M}}}{\sqrt{T}}}_* \leq \sqrt{s}\smalliii{\Delta_{B|\mathbb{S}} }_{\F} + \frac{\lambda_\Theta}{\lambda_B} (\sqrt{2K}) \smalliii{\tfrac{\Delta_{\Theta|\mathbb{M}}}{\sqrt{T}}}_{\F} \\
	& \leq \sqrt{s}\smalliii{\Delta_B}_{\F} + \frac{\lambda_\Theta}{\lambda_B} (\sqrt{2K}) \smalliii{\Delta_{\Theta}/\sqrt{T}}_{\F}.
	\end{split}
	\end{equation}
	Plug~\eqref{eqn:Rup} into~\eqref{eqn:B1up}, and by the Cauchy-Schwartz inequality, we have
	\begin{equation}\label{eqn:DeltaB1}
	\smallii{\Delta_B}^2_1\leq 16\big( s + (2K)(\lambda_\Theta/\lambda_B)^2\big) \Big[\smalliii{\Delta_B}_{\F}^2 + \smalliii{\Delta_\Theta/\sqrt{T}}_{\F}^2\Big] + 16\tau_T\|B^\star_{\mathbb{S}^c}\|^2_1.
	\end{equation}
	Combine~\eqref{eqn:LHSlow1} and~\eqref{eqn:DeltaB1}, a lower bound for the LHS of~\eqref{eqn:basic2} is given by 
	\begin{equation*}
	\Big[\frac{\alpha_{\text{RSC}}}{2} - 16\tau_T\big(s + (2K)(\tfrac{\lambda_\Theta}{\lambda_B})^2 \big)\Big]\smalliii{\Delta_B}_{\F}^2 + \Big[\frac{1}{2}-  16\tau_T\big(s + (2K)(\tfrac{\lambda_\Theta}{\lambda_B})^2 \big) \Big] \smalliii{\Delta_\Theta/\sqrt{T}}_{\F}^2 - 16 \tau_T\|B^\star_{\mathbb{S}^c}\|^2_1.
	\end{equation*}
	With the designated choice of $\tau_T$ satisfying $16\tau_T\big(s + (2K)(\tfrac{\lambda_\Theta}{\lambda_B})^2 \big) \leq \min\{\alpha_{\text{RSC}},1\}/4$, the above bound can be further lower bounded~by 
	\begin{equation}\label{eqn:LHSlow2}
	\frac{\min\{\alpha_{\text{RSC}},1\}}{4} \Big(\smalliii{\Delta_B}_{\F}^2 + \smalliii{\Delta_\Theta/\sqrt{T}}_{\F}^2  \Big) - 16\tau_T \|B^\star_{\mathbb{S}^c}\|^2_1.
	\end{equation}
	
	\paragraph{Part (ii).} Next, we obtain an upper bound for the RHS of~\eqref{eqn:basic2}.  Using the triangle inequality and H\"{o}lder's inequality, the first term satisfies
	\begin{equation}\label{eqn:firstterm}
	\begin{split}
	\frac{1}{T}|\llangle \mathbf{X}_{T-1}\Delta_B^\top,\widehat{\Theta}-\Theta^\star\rrangle| &\leq \frac{1}{T}|\llangle \Delta_{B}^\top,\mathbf{X}_{T-1}^\top\widehat{\Theta}\rrangle| + \frac{1}{T}|\llangle \Delta_{B}^\top,\mathbf{X}_{T-1}^\top\Theta^\star\rrangle| \\
	&\leq  \|\Delta_{B}\|_1 \|\mathbf{X}_{T-1}^\top\widehat{\Theta}/T\|_\infty +  \|\Delta_{B}\|_1  \|\mathbf{X}_{T-1}^\top\Theta^\star/T\|_\infty\\
	& \leq \|\Delta_{B}\|_1 \cdot\smalliii{\mathbf{X}_{T-1}/T}_1\cdot \|\widehat{\Theta}\|_\infty + \|\Delta_{B}\|_1 \cdot\smalliii{\mathbf{X}_{T-1}/T}_1\cdot \|\Theta^\star\|_\infty.
	\end{split}
	\end{equation}
	Using the fact that both $\Theta^\star$ and $\widehat{\Theta}$ are feasible and satisfy the box constraint $\Theta\in\mathbb{B}_\infty(\phi,\mathbf{X}_{T-1})$, the RHS of~\eqref{eqn:firstterm} is upper bounded by $\tfrac{2\phi}{\sqrt{Tp}}\cdot \|\Delta_{B}\|_1$; thus, by choosing $\lambda_B\geq 4\phi/\sqrt{Tp}$,
	we have
	\begin{equation*}
	\frac{1}{T}|\llangle \mathbf{X}_{T-1}\Delta_B^\top,\widehat{\Theta}-\Theta^\star\rrangle| \leq  \frac{\lambda_B}{2}\smallii{\Delta_B}_1.
	\end{equation*}
	With~\eqref{eqn:holder1} and~\eqref{eqn:holder2}, by choosing $\lambda_B \geq 2\|\mathbf{X}_{T-1}^\top \mathbf{E}/T\|_\infty + 4\phi/\sqrt{Tp}$ and $\lambda_\Theta\geq \Lambda^{1/2}_{\max}(S_{\mathbf{E}})$, the following upper bound holds for the RHS of~\eqref{eqn:basic2}:
	{\small
		\begin{align}
		&\lambda_B \|\Delta_B\|_1 + \lambda_\Theta\|\Delta_{\Theta}\|_*  +  \lambda_B\big( \mathcal{R}(\Delta_{B|\mathbb{S}},\Delta_{\Theta|\mathbb{M}}) - \mathcal{R}(\Delta_{B|\mathbb{S}^c},\Delta_{\Theta|\mathbb{M}^\perp}) + 2\mathcal{R}(B^\star_{|\mathbb{S}^c}, \Theta^\star_{|\mathbb{M}^\perp}) \big) \notag \\
		\stackrel{(1)}{\leq} ~ & \lambda_B\big( \|\Delta_{B|\mathbb{S}}\|_1 + \|\Delta_{B|\mathbb{S}^c}\|_1\big) + \lambda_\Theta\big( \smalliii{\tfrac{\Delta_{\Theta|\mathbb{M}}}{\sqrt{T}}}_* +  \smalliii{\tfrac{\Delta_{\Theta|\mathbb{M}^\perp}}{\sqrt{T}}}_*\big) \notag\\
		& + \lambda_B\big( \mathcal{R}(\Delta_{B|\mathbb{S}},\Delta_{\Theta|\mathbb{M}}) - \mathcal{R}(\Delta_{B|\mathbb{S}^c},\Delta_{\Theta|\mathbb{M}^\perp}) \big)+ 2\lambda_B\mathcal{R}(B^\star_{\mathbb{S}^c}, \Theta^\star_{\mathbb{M}^\perp}) \notag \\
		\stackrel{(2)}{\leq} ~ & 2\lambda_B \|\Delta_{B|\mathbb{S}}\|_1 + 2\lambda_\Theta\smalliii{\tfrac{\Delta_{\Theta|\mathbb{M}}}{\sqrt{T}}}_* + 2\lambda_B\mathcal{R}(B^\star_{\mathbb{S}^c}, \Theta^\star_{\mathbb{M}^\perp}) = 2\lambda_B\mathcal{R}(\Delta_{B|\mathbb{S}},\Delta_{\Theta|\mathbb{M}}) + 2\lambda_B\mathcal{R}(B^\star_{\mathbb{S}^c}, \Theta^\star_{\mathbb{M}^\perp})  \notag\\
		\stackrel{(3)}{\leq} ~ &  (2\lambda_B)\sqrt{s}\smalliii{\Delta_B}_{\F} + (2\lambda_\Theta)\sqrt{2K}\smalliii{\Delta_{\Theta}/\sqrt{T}}_{\F} + (2\lambda_B)\|B^\star_{\mathbb{S}^c}\|_1 \label{eqn:RHSup1}
		\end{align}
	}%
	\normalsize
	where (1) uses~\eqref{eqn:Delta1} and~\eqref{eqn:Delta2}; (2) is obtained by writing out $\mathcal{R}(\cdot,\cdot)$ and canceling  terms; and (3) uses~\eqref{eqn:Rup}. 
	\paragraph{Part (iii).} Aligning~\eqref{eqn:LHSlow2} and~\eqref{eqn:RHSup1} then rearranging terms associated with $\Delta_B$ and $\Delta_\Theta$ gives the claimed bound in~\eqref{eqn:eb1}.
\end{proof}

%%%%%%%%%%%%%%%%%%%%%%%%%%%%%%%%%%%%%%%%%%%%%%%%%%%%%%
%***********************************************
% PROOF OF Corollary 1
%***********************************************
%%%%%%%%%%%%%%%%%%%%%%%%%%%%%%%%%%%%%%%%%%%%%%%%%%%%%%
\begin{proof}[Proof of Corollary~\ref{cor:eb}]
	First we note that by the definition of $\mathbb{B}_q(R_q)$, the following holds for the strong support set
	\begin{equation*}\label{eqn:strong-support}
	R_q \geq \sum_{i,j} |B_{ij}|^q \geq \sum_{(i,j)\in S^\star_\eta} |B_{ij}|^q \geq \eta^q s_\eta,
	\end{equation*}
	which then gives $\eta^{-q}R_q$. Further, the following inequality holds for the weak support set:
	\begin{equation*}	\label{eqn:weak-support}
	\sum_{(i,j)\notin S^\star_\eta} |B_{ij}| = \sum_{(i,j)\notin S^\star_\eta} (|B_{ij}|^{q}|B_{ij}|^{1-q}) \leq R_q \eta^{1-q}.
	\end{equation*}
	Setting $\eta = \lambda_B/\alpha'$ and plugging~\eqref{eqn:strong-support} and~\eqref{eqn:weak-support} into~\eqref{eqn:eb1} yields the desired result.
\end{proof}

\subsection{The $\sin\theta$ Error.}
Theorem~\ref{thm:fix-bound-convex} establishes the error bound for the estimated factor hyperplane through the quantity $\Delta_\Theta$. Here we provide an account for the error bound of the space spanned by the estimated latent factors relative to the true underlying one. To that end we derive an error bound for $\sin\theta$ that measures the distance between the estimated factor space and the true factor space. In particular, we focus on analyzing the error between the leading rank-$K$ subspace spanned by $\Theta$ and $\widehat{\Theta}$, although potentially $\widehat{\Theta}$ could span an $r$-dimensional subspace (whenever $r\neq K$) that depends on the value of the selected $\lambda_\Theta$.  

First we note that to examine the $\sin\theta$ error of the estimated factor space is equivalent to examine the $\sin\theta$ distance between $\widehat{U}_K$ and $U^\star_K$, where $\widehat{U}_K$ and $U^\star_K$ are the first $K$ left singular vectors corresponding to $\widehat{\Theta}$ and $\Theta^\star$, respectively. Specifically, the angle between the spaces they span is defined as
\begin{equation}\label{eqn:sintheta}
\theta(\widehat{\mathbf{F}}_K,\mathbf{F}) = \theta(\widehat{U}_K,U^\star_K) := \text{diag}\Big( \cos^{-1} (\bar{\sigma}_1),\cos^{-1}(\bar{\sigma}_2),\dots,\cos^{-1}(\bar{\sigma}_K) \Big), 
\end{equation}
where $\bar{\sigma}_1\geq\bar{\sigma}_2\geq\dots\geq\bar{\sigma}_K\geq 0$ are singular values of $\widehat{U}_K^\top U^\star_K$. The following proposition associates the error of $\sin\theta$ to that of $\Delta_\Theta$.

%% STATEMENT OF THE PROPOSITION %%
\begin{proposition}[$\sin\theta$ error of the estimated factor space]\label{prop:space-error} Suppose the estimated factor hyperplane $\widehat{\Theta}\in\R^{n\times p}$ is obtained by solving~\eqref{eqn:optimization0}, whose error is given by $\Delta_{\Theta}=\widehat{\Theta}-\Theta^\star$. Let $\sigma_1$ and $\sigma_K$ be the leading and the smallest nonzero singular values of $\Theta^\star$. The following bound holds for the $\sin\theta$ distance between the estimated and the true factor spaces: 
	\begin{equation}\label{eqn:eb-factorspace}
	\smalliii{\sin\theta(\widehat{\mathbf{F}}_K,\mathbf{F})}_{\F}^2 \leq \frac{2(2\sigma_1 + \smalliii{\Delta_\Theta}_{\op})\min\left\{\sqrt{K}\smalliii{\Delta_\Theta}_{\op},\smalliii{\Delta_\Theta}_{\F}\right\}  }{\sigma_K^2}.
	\end{equation}	
\end{proposition}

The bound in~\eqref{eqn:eb-factorspace} is obtained by considering $\widehat{\Theta}$ as a $\Delta_\Theta$-perturbation of $\Theta^\star$, and the size of the perturbation is upper bounded in Frobenius norm given by $\smalliii{\Delta_{\Theta}/\sqrt{n}}_{\F}$ given in Theorem~\ref{thm:fix-bound-convex}. The stronger the minimum signal is for the true space (i.e., $\sigma_K$), the tighter the $\sin\theta$ error bound will be. Note that for the true space spanned by $\mathbf{F}$, although it is not observable, it can nevertheless be interpreted as a random (but fixed for this specific part of the analysis) realization drawn from the specified VAR model driving the dynamics of $F_t$, which in turn directly influences the evolution of the observable $X_t$ process.

%%%%%%%%%%%%%%%%%%%%%%%%%%%%%%%%%%%%%%%%%%%%%%%%%%%%%%
%***********************************************
% PROOF OF PROPORSITION 1
%***********************************************
%%%%%%%%%%%%%%%%%%%%%%%%%%%%%%%%%%%%%%%%%%%%%%%%%%%%%%
\begin{proof}[Proof of Proposition~\ref{prop:space-error}]
	First, we note that for any given $\widehat{\Theta}=\Theta^\star + \Delta_{\Theta}$, it can be viewed as a $\Delta_\Theta$-perturbation with respect to the true $\Theta^\star$. As mentioned in the main text, as invertible linear transformations preserve the subspace, so does scaling (with a non-zero scale factor), it is equivalent to examining the $\sin\theta$ distance between the first $K$ singular vectors of $\widehat{\Theta}$ and $\Theta^\star$ (denoted by $\widehat{U}$ and $U^\star$, resp.). The rest follows seamlessly from the perturbation theory of singular vectors. Specifically, by applying \citet[Theorem 3]{yu2015useful} and assuming the singular values of $\Theta^\star$ are given by $\sigma_1>\sigma_2>\dots>\sigma_{K}>\sigma_{K+1}=\dots=\sigma_{\min\{T,p\}}=0$, the following bound holds for $\smalliii{\sin(\widehat{U},U^\star)}$:
	\begin{equation*}
	\smalliii{\sin\theta(\widehat{U},U^\star)}_{\F}^2 \leq \frac{2(2\sigma_1 + \smalliii{\Delta_\Theta}_{\op})\min\left\{\sqrt{K}\smalliii{\Delta_\Theta}_{\op},\smalliii{\Delta_\Theta}_{\F}\right\}  }{\sigma_K^2}.
	\end{equation*}
	Note that the same bound holds for the $\sin\theta$ distance between the factor spaces. 
\end{proof}

%%%%%%%%%%%%%%%%%%%% Proof of Lemmas %%%%%%%%%%%%%%%%%%%%%%%%%
\section{Proofs for Lemmas. }\label{appendix:proof-of-theory-lemma}

%%%%%%%%%%%%%%%%%%%%%%%%%%%%%%%%%%%%%%%%%%%%%%%%%%%%%%
%***********************************************
% PROOF OF LEMMA 1
%***********************************************
%%%%%%%%%%%%%%%%%%%%%%%%%%%%%%%%%%%%%%%%%%%%%%%%%%%%%%
\begin{proof}[Proof of Lemma~\ref{lemma:RSC}]
	First, suppose we have
	\begin{equation}\label{eqn:RSC-stronger}
	\frac{1}{2}v' S_{\mathbf{X}} v = \frac{1}{2}v'\Big( \frac{\mathbf{X}'\mathbf{X}}{T} \Big)v\geq \frac{\alpha_{\text{RSC}}}{2} \|v\|_2^2 - \tau_T\|v\|_1^2,\quad \forall~v\in\R^p;
	\end{equation}
	then, for all $\Delta\in\R^{p\times p}$, and letting $\Delta_j$ denote its $j$th column, the RSC condition automatically holds since
	\begin{equation*}
	\tfrac{1}{2T}\smalliii{\mathbf{X}\Delta}_{\F}^2 = \frac{1}{2}\sum_{j=1}^q \Delta_j'\big( \tfrac{\mathbf{X}'\mathbf{X}}{T} \big) \Delta_j \geq \frac{\alpha_{\text{RSC}}}{2} \sum_{j=1}^q \|\Delta_j\|_2^2 - \tau_T \sum_{j=1}^q \|\Delta_j\|_1^2 \geq \frac{\alpha_{\text{RSC}}}{2}\smalliii{\Delta}_{\F}^2 - \tau_T \|\Delta\|_1^2.
	\end{equation*}
	Therefore, it suffices to verify that~\eqref{eqn:RSC-stronger} holds. In \citet[Proposition 4.2]{basu2015estimation}, the authors prove a similar result under the assumption that $X_t$ is a $\text{VAR}(d)$ process. Here, we adopt the same proof strategy and state the result for a {\em more general process} $X_t$. 
	
	Specifically, by \citet[Proposition 2.4(a)]{basu2015estimation}, $\forall v\in\R^p,\|v\|\leq 1$ and $\eta >0$,
	\begin{equation*}
	\PP \Big[ \big| v'\big(S_{\mathbf{X}}-\Gamma_X(h)\big) v \big| >  2\pi\mathcal{M}(g_X)\eta \Big]\leq 2\eta \exp \Big(-cT\min\{\eta^2,\eta\} \Big). 
	\end{equation*}
	Applying the discretization in \citet[Lemma F.2]{basu2015estimation} and taking the union bound, define $\mathbb{K}(2s):=\{v\in\R^p, \|v\|\leq 1, \|v\|_0\leq 2k\}$, and the following inequality holds:
	\begin{equation*}
	\PP \Big[ \sup\limits_{v\in\mathbb{K}(2k)}\big| v'\big(S_{\mathbf{X}}-\Gamma_X(h)\big) v \big| >  2\pi\mathcal{M}(g_X)\eta \Big]\leq 2\exp\Big( -cT\min\{\eta,\eta^2\} + 2k \min\{\log p, \log (21 ep/2k) \} \Big). 
	\end{equation*}
	With the specified $\gamma=54\mathcal{M}(g_X)/\mathfrak{m}(g_X)$, set $\eta=\gamma^{-1}$, then apply results from \citet[Lemma 12]{loh2012high} with $\Gamma=S_{\mathbf{X}}-\Gamma_X(0)$ and $\delta=\pi\mathfrak{m}(g_X)/27$, so that the following holds
	\begin{equation*}
	\frac{1}{2}v' S_{\mathbf{X}} v \geq \frac{\alpha_{\text{RSC}}}{2}\|v\|^2 - \frac{\alpha_{\text{RSC}}}{2k}\|v\|_1^2,
	\end{equation*}
	with probability at least $1-2\exp\big(-cT\min\{\gamma^{-2},1\} + 2k\log p \big)$ and note that $\min\{\gamma^{-2},1\}=\gamma^{-2}$ since $\gamma>1$. Finally, let $k= \min\{cT\gamma^{-2}/(c'\log p),1\}$ for some $c'>2$, and conclude that with probability at least $1-c_1\exp(-c_2T)$, the inequality in~\eqref{eqn:RSC-stronger} holds with 
	\begin{equation*}
	\alpha_{\text{RSC}} = \pi\mathfrak{m}(g_X), \qquad \tau_T = \alpha_{\text{RSC}}\gamma^2\frac{\log p}{2T},
	\end{equation*}
	and so does also the RSC condition. 
\end{proof}

%%%%%%%%%%%%%%%%%%%%%%%%%%%%%%%%%%%%%%%%%%%%%%%%%%%%%%
%***********************************************
% PROOF OF LEMMA 2
%***********************************************
%%%%%%%%%%%%%%%%%%%%%%%%%%%%%%%%%%%%%%%%%%%%%%%%%%%%%%
\begin{proof}[Proof of Lemma~\ref{lemma:boundXE}]
	We note that
	\begin{equation*}
	\frac{1}{T}\smallii{\mathbf{X}_{T-1}^\top \mathbf{E}}_\infty = \max\limits_{1\leq i,j\leq p} \big|e_i^\top \big(\mathbf{X}_{T-1}^\top \mathbf{E}/n\big)e_j\big|,
	\end{equation*}
	where $e_i$ is the $p$-dimensional standard basis with the $i$th entry being 1. Applying \citet[Proposition 2.4(b)]{basu2015estimation}, for an arbitrary pair of $(i,j)$, the following inequality holds:
	\begin{equation*}
	\PP\Big[ \big|e_i^\top \big(\mathbf{X}_{T-1}^\top \mathbf{E}/T\big)e_j \big| > 2\pi \big(\mathcal{M}(g_X) + \mathcal{M}(g_\epsilon) + \mathcal{M}(g_{X,\widetilde{\epsilon}}) \big)\eta\Big] \leq 6\exp\Big( -cT \min\{\eta^2,\eta\}\Big). 
	\end{equation*}	
	Take the union bound over all $1\leq i,j \leq p$, and the following bound holds:
	\begin{equation*}
	\PP \Big[ \max\limits_{1\leq i,j\leq p} \big|e_i^\top \big(\mathbf{X}_{T-1}^\top \mathbf{E}/T\big)e_j\big| > 2\pi \big(\mathcal{M}(g_X) + \mathcal{M}(g_\epsilon) + \mathcal{M}(g_{X,\widetilde{\epsilon}}) \big)\eta\Big] \leq 6\exp \Big( -cn \min\{\eta^2,\eta\} + 2\log p\Big).
	\end{equation*}
	Set $\eta=c'\sqrt{\log p/T}$ for $c'>(2/c)$ and with the choice of $T\succsim \log p$, $\min\{\eta^2,\eta\} = \eta^2$, then with probability at least $1-c_1\exp(-c_2\log p)$, the following bound holds:
	\begin{equation*}
	\frac{1}{T}\vertii{\mathbf{X}_{T-1}^\top \mathbf{E}}_\infty\leq c_0 \big(\mathcal{M}(g_X) + \mathcal{M}(g_\epsilon) + \mathcal{M}(g_{X,\widetilde{\epsilon}}) \big)\sqrt{\frac{\log p}{T}}.
	\end{equation*}
\end{proof}

%%%%%%%%%%%%%%%%%%%%%%%%%%%%%%%%%%%%%%%%%%%%%%%%%%%%%%
%***********************************************
% PROOF OF LEMMA Emax
%***********************************************
%%%%%%%%%%%%%%%%%%%%%%%%%%%%%%%%%%%%%%%%%%%%%%%%%%%%%%
\begin{proof}[Proof of Lemma~\ref{lemma:Emax}]
	For $\mathbf{E}$ whose rows are iid realizations of a sub-Gaussian random vector $\epsilon_t$, by \citet[Lemma 9]{wainwright2009sharp}, the following bound holds:
	\begin{equation*}
	\mathbb{P}\Big[ \vertiii{S_{\mathbf{E}} - \Sigma_\epsilon }_{op}\geq \Lambda_{\max}(\Sigma_\epsilon) \delta(T,p,\eta)\Big]\leq 2\exp(-T\eta^2/2), 
	\end{equation*}
	where $\delta(T,p,\eta):=2\big(\sqrt{\frac{p}{T}}+\eta\big)+\big(\sqrt{\frac{p}{T}}+\eta\big)^2$. In particular, by triangle inequality, with probability at least $1-2\exp(-T\eta^2/2)$, 
	\begin{equation*}
	\smalliii{S_{\mathbf{E}}}_{op} \leq \smalliii{\Sigma_\epsilon}_{op} + \smalliii{S_{\mathbf{E}}-\Sigma_\epsilon}_{op} \leq \Lambda_{\max}(\Sigma_\epsilon) + \Lambda_{\max}(\Sigma_\epsilon) \delta(T,p,t).
	\end{equation*}
	So for $T\geq p$, by setting $\eta=1$, which yields $\delta(T,p,\eta)\leq 8$ so that with probability at least $1-2\exp(-T/2)$, the following bound holds:
	\begin{equation*}
	\Lambda_{\max}(S_{\mathbf{E}}) \leq 9 \Lambda_{\max}(\Sigma_\epsilon).
	\end{equation*}
\end{proof}

%%%%%%%%%%%%%%%%%%%%%%%%%%%%%%%%%%%%%
\section{Proofs of Auxiliary Lemmas.}\label{appendix:auxLemma}

In this section, proofs of auxiliary lemmas~\ref{auxLemma:triangle}, \ref{auxLemma:rank} are provided. Variations of these Lemmas have been proved in \citet{negahban2012unified} and \citet{lin2017regularized}; nevertheless, we provide them also here for the sake of completeness.

\begin{proof}[Proof of Lemma~\ref{auxLemma:triangle}]
	With the definition of~\eqref{eqn:projS} and~\eqref{eqn:projM} and that the $\ell_1$ norm and nuclear norm are decomposable, the following bound holds: 
	\begin{equation*}
	\mathcal{R}(B^\star,\Theta^\star) = \smallii{B^\star_{\mathbb{S}} + B^\star_{\mathbb{S}^c}}_1 + \frac{\lambda_\Theta}{\lambda_B}\smalliii{\Theta^\star_{\mathbb{M}} + \Theta^\star_{\mathbb{M}^\perp}}_* = 
	\mathcal{R}(B^\star_\mathbb{S}, \Theta^\star_{\mathbb{M}}) + \mathcal{R}(B^\star_{\mathbb{S}^c}, \Theta^\star_{\mathbb{M}^\perp}).
	\end{equation*}
	It then follows that 
	\begin{equation*}
	\begin{split}
	\mathcal{R}(\widehat{B},\widehat{\Theta}) &= \mathcal{R}(B^\star + \Delta_B,\Theta^\star + \Delta_\Theta) \\
	& = \smallii{B^\star_{\mathbb{S}} + B^\star_{\mathbb{S}^c} + \Delta_{B|\mathbb{S}} + \Delta_{B|\mathbb{S}^c}}_1 + \frac{\lambda_\Theta}{\lambda_B}\smalliii{\Theta^\star_{\mathbb{M}} + \Theta^\star_{\mathbb{M}^\perp} + \Delta_{\Theta|\mathbb{M}} + \Delta_{\Theta|\mathbb{M}^\perp} }_* \\
	& \geq \smallii{B^\star_{\mathbb{S}} + \Delta_{B|\mathbb{S}^c}}_1 - \smallii{\Delta_{B|\mathbb{S}}}_1 + \frac{\lambda_\Theta}{\lambda_B}\Big( \smalliii{\Theta^\star_{\mathbb{M}}+ \Delta_{\Theta|\mathbb{M}^\perp}}_* - \smalliii{\Delta_{\Theta|\mathbb{M}}}_* \Big)  - \mathcal{R}(B^\star_{\mathbb{S}^c}, \Theta^\star_{\mathbb{M}^\perp}) \\
	& \stackrel{(i)}{\geq} \smallii{B^\star_{\mathbb{S}}}_1 + \smallii{\Delta_{B|\mathbb{S}^c}}_1 - \smallii{\Delta_{B|\mathbb{S}}}_1 + \frac{\lambda_\Theta}{\lambda_B}\Big( \smalliii{\Theta^\star_{\mathbb{M}}}_*+ \smalliii{\Delta_{\Theta|\mathbb{M}^\perp}}_* - \smalliii{\Delta_{\Theta|\mathbb{M}}}_* \Big)  - \mathcal{R}(B^\star_{\mathbb{S}^c}, \Theta^\star_{\mathbb{M}^\perp})\\
	& \geq \mathcal{R}(B^\star_{\mathbb{S}}, \Theta^\star_{\mathbb{M}}) + \mathcal{R}(\Delta_{B|\mathbb{S}^c}, \Delta_{\Theta|\mathbb{M}^\perp}) - \mathcal{R}(\Delta_{B|\mathbb{S}}, \Delta_{\Theta|\mathbb{M}}) - \mathcal{R}(B^\star_{\mathbb{S}^c}, \Theta^\star_{\mathbb{M}^\perp}) \\
	& = \mathcal{R}(B^\star, \Theta^\star) + \mathcal{R}(\Delta_{B|\mathbb{S}^c}, \Delta_{\Theta|\mathbb{M}^\perp}) - \mathcal{R}(\Delta_{B|\mathbb{S}}, \Delta_{\Theta|\mathbb{M}}) - 2\mathcal{R}(B^\star_{\mathbb{S}^c}, \Theta^\star_{\mathbb{M}^\perp}).
	\end{split}
	\end{equation*}
	where (i) uses the property of decomposable regularizers. By rearranging, we obtain the desired inequality. 
\end{proof}

\begin{proof}[Proof of Lemma~\ref{auxLemma:rank}]
	Let the SVD of $\Theta^\star$ be given by $\Theta^\star = (U^\star) D (V^\star)^\top$, where both $U^\star$ and $V^\star$ are orthogonal matrices. Assume $\text{rank}(\Theta^\star) = K$. For $\Delta\in\R^{T\times p}$, define $\widetilde{\Delta}$ as below and it is partitioned as:
	\begin{equation*}
	\widetilde{\Delta} := (U^\star)^\top \Delta (V^\star) = \begin{bmatrix}
	\widetilde{\Delta}_{11} & \widetilde{\Delta}_{12} \\ \widetilde{\Delta}_{21} & \widetilde{\Delta}_{22}
	\end{bmatrix}, \qquad \text{where}~~\widetilde{\Delta}_{11}\in\R^{K\times K}.
	\end{equation*}
	Then by further defining 
	\begin{equation*}
	\Delta_{\mathbb{M}} := U^\star \begin{bmatrix}
	\widetilde{\Delta}_{11} & \widetilde{\Delta}_{12} \\ \widetilde{\Delta}_{21} & O
	\end{bmatrix}(V^\star)^\top \qquad \text{and} \qquad \Delta_{\mathbb{M}^\perp} := U^\star \begin{bmatrix}
	O & O \\ O & \widetilde{\Delta}_{22}
	\end{bmatrix}(V^\star)^\top,
	\end{equation*}
	it is straightforward to see that $\Delta_{\mathbb{M}} + \Delta_{\mathbb{M}^\perp} = \Delta$. Moreover, 
	\begin{equation*}
	\text{rank}(\Delta_{\mathbb{M}}) \leq \text{rank}\left(U^\star \begin{bmatrix}
	\widetilde{\Delta}_{11} & \widetilde{\Delta}_{12} \\ O & O
	\end{bmatrix}(V^\star)^\top \right) + \text{rank}\left(U \begin{bmatrix}
	\widetilde{\Delta}_{11} & O \\ \widetilde{\Delta}_{21} & O
	\end{bmatrix}(V^\star)^\top \right) \leq 2K.
	\end{equation*}
\end{proof}

%%%%%%%%%%%%%%%%%%% ID %%%%%%%%%%%%%%%%%%%%%%%
\section{Model identifiability considerations.}\label{appendix:modelID}

In this section, we provided a detailed account of the model identifiability issue, due to the factor space being latent. 

Consider the full identification of the given model $X_t = \Lambda F_t + BX_{t-1} + \epsilon_t$. 
Similar to the analysis in \citet{bai2016estimation} for the informational series of a factor-augmented VAR model, there exist invertible matrices $M_{11}\in\R^{K\times K}$ and $M_{12}\in\R^{K\times p}$ such that 
\begin{equation}\label{eqn:id}
X_t = \Lambda F_t + BX_{t-1} + \epsilon_t = \underbrace{(\Lambda M_{11})}_{\bar{\Lambda}} \underbrace{ (M_{11}^{-1}F_t - M_{11}^{-1}M_{12}X_{t-1})}_{\bar{F}_t} + \underbrace{(B + M_{12})}_{\bar{B}}X_{t-1} + \epsilon_t,
\end{equation}
which are observationally equivalent to the original model. So for the model to be fully identifiable (including the factors), a total number of $K^2+Kp$ restrictions is required. If exact identification of the factors is not required, then $Kp$ restrictions are required to separate the space spanned by $F_t$ from that by $X_{t-1}$. In low dimensional settings with a different model setup, an estimation procedure based on~\eqref{eqn:id} that takes into consideration these $Kp$ restrictions can be carried out in a two step procedure involving estimation followed by rotation, with the aid of $(\mathbf{X}_{T-1}^\top \mathbf{X}_{T-1})^{-1}$ and the associated orthogonal projection operator \citep[see][for details]{bai2016estimation}. In the high-dimensional setting, however, neither auxiliary quantity is properly defined, and hence the above strategy can not be operationalized. As oppose to imposing additional model assumptions that would be stringent and only made for the sake of mathematical convenience, we incorporate these $Kp$ restrictions implicitly and approximately, through the assumption that the amount of interaction between the latent factor space and the lag space is controlled, which manifests itself in the technical developments as the product of the total signal present in these two spaces. In the formulated optimization problem, with properly selected tuning parameters, the global minimizer of the convex program exhibits good statistical behavior in terms  of its error that does {\em not grow} with $p$ or $n$, so that there is adequate control over the performance of the estimator, even though this upper bound of the error does not vanish asymptotically. This represents the price to be paid for handling strongly correlated idiosyncratic components in approximate factor models, under {\em minimal} identifiability restrictions.

\begin{remark}[Illustration on the additional box constraint]\label{remark:boxconstraint}
	In the same spirit as \citet{agarwal2012noisy}, to distinguish the low rank hyperplane $\Theta$ from the lagged space from a theoretical standpoint , we restrict $\Theta$ to be in the constrained set $\varphi_{\mathcal{R}}(\Theta)\leq \phi$, where
	\begin{equation*}
	\varphi_{\mathcal{R}}(\Theta) := \kappa(\mathcal{R}^*)\mathcal{R}^*(\Theta) \smalliii{\mathbf{X}/\sqrt{T}}_{\text{op}}, \qquad \kappa(\mathcal{R}^*) := \sup\limits_{\Theta\neq 0} \big(\frac{\smalliii{\Theta}_\F}{\mathcal{R}^*(\Theta)}\big).
	\end{equation*}
	$\mathcal{R}^*$ is the dual norm of some regularizer $R$. Note that the product $\kappa(\mathcal{R}^*)\mathcal{R}^*(\Theta)$ measures the spikiness of $\Theta$ w.r.t. $\mathcal{R}$; in the setup of interest in this paper, it corresponds to the $\ell_1$-norm regularizer associated with the sparse component; hence $\mathcal{R}^*(\Theta) = \|\Theta\|_\infty$ and $\kappa(\mathcal{R}^*) = \sqrt{Tp}$. This constrained set leads to the box constraint given in~\eqref{eqn:optimization0}.
\end{remark}

\begin{remark}\label{remark:random} 
	For approximate factor models, a large panel size (large $p$) is helpful, since the estimated factors are obtained through cross-sectional aggregation. In particular, as discussed in \citet{chamberlain1983arbitrage} and subsequent work, by assuming that the leading $K$ eigenvalues of $\Sigma_X$ diverge, whereas all eigenvalues of $\Sigma_u$ are bounded, separation between the common factors and the idiosyncratic components is achieved as the panel size $p$ goes to infinity. On the other hand, the Stock-Watson formulation \citep{stock2002forecasting} adopted in our work which accounts explicitly for strong correlations amongst the coordinates of the idiosyncratic component, leads to a high-dimensional sparse regression modeling framework. Hence, the estimates for the time-lags of the $X_t$ process suffer from the curse of dimensionality, if we do not compensate appropriately by an increase in the sample size. Hence, we need to strike a balance between these two competing
	forces. Specifically, when updating the estimate of the factor hyperplane by aggregating cross-sectional information and compress it to a subspace with reduced dimension through the SVD, a larger panel $p$ is helpful. On the other hand, when updating the estimate of the sparse transition matrix, a very
	high $p$ is detrimental, unless appropriately compensated by a larger sample size $n$. In addition, the temporal dependence of the coordinates of the	$X_t$ process along with the presence of the latent factors add further complications. Thus, careful balancing of these competing issues is needed to obtain estimates of the model parameters with adequate error control. 
\end{remark}

\mbox{}
\vfill
\begin{center}
{\large\bf SUPPLEMENTARY MATERIAL}
\end{center}
\begin{description}
\item[Supplement to $d$-lag dependence.] Generalization to $\VAR(d)$ dependence.
\end{description}

\titleformat{\section}{\large\normalfont\centering\bfseries}{}{0.5em}%{\uppercase}
\clearpage
\setcounter{page}{1}
\pagenumbering{arabic}
\section{Generalization to {$\VAR(d)$} Dependence.}\label{appendix:VARd}

The proposed modeling framework and estimation procedure are easily generalizable to cases where the idiosyncratic error $u_t$ exhibits further into the past temporal dependence, and come with similar theoretical guarantees. Specifically, we use a sparse $\VAR(d)$ model to account for such dependency, that is,
\begin{equation*}
\mathcal{B}_d(L)u_t = \epsilon_t, \qquad \text{where}\qquad \mathcal{B}_d(L) = \I_p - B_1L - \dots - B_dL^d, \quad \epsilon_t\stackrel{\text{i.i.d}}{\sim}\mathcal{N}(0,\sigma_\epsilon^2\I_p).
\end{equation*}
with $B_k$'s assumed sparse. By stacking the lagged values of the factors and the corresponding loading matrices, the dynamic of the observable process $X_t$ can be written in the following form, in terms of the latent static factor $F_t\in\R^K$:
\begin{equation}\label{model:ours-VARd}
X_t = \Lambda F_t + B_1 X_{t-1} + B_2 X_{t-2} + \cdots + B_dX_{t-d} + \epsilon_t.
\end{equation}
Similar to the $\VAR(1)$ case, the condition required for stationarity is the same as cases where $X_t$ were a $\VAR(d)$ process, that is, all roots of $\det(\mathcal{B}_d(z))$ should lie outside the unit circle: $\det(\mathcal{B}_d(z))\neq0$ for all $|z|\leq1$. Note that with the model specification in~\eqref{model:ours-VARd}, the spectral density of $X_t$ takes the following form:
\begin{equation*}
g_X(\omega) = \big[ \mathcal{B}_d^{-1}(e^{i\omega})\big]\Big(\Lambda g_F(\omega)\Lambda^\top + g_\epsilon(\omega) + g_{\epsilon,F}(\omega)\Lambda^\top + \Lambda g_{F,\epsilon}(\omega)\Big)\big[ \mathcal{B}_d^{-1}(e^{i\omega})\big]^*.
\end{equation*} 

\paragraph{Estimation and theoretical guarantees.} Given a snapshot of the realizations from $\{X_t\}$, denoted by $\{x_0,\dots,x_T\}$, we can estimate $\{B_k,k=1,\dots,d\}$ and the factor hyperplane in an analogous way. Specifically, let the contemporaneous response and the lagged predictor matrices be $\mathbf{X}^d_{T}\in \R^{T_d\times p}$ and $\mathbf{X}^d_{T-1}\in \R^{T_d\times dp}$ by stacking the observations in their rows with $T_d = T-d+1$ being the sample size. $\mathbf{E}^d_n$ is similarly defined to $\mathbf{X}^d_{n}$. Further, letting $B:=[B_1,B_2,\dots,B_d]\in\R^{p\times dp}$, then with $\mathbf{F}$ and $\Theta$ identically defined to those in Section~\ref{sec:estimation}, we can write 
\begin{equation*}
\mathbf{X}_T^d = \mathbf{F}\Lambda^\top + \mathbf{X}_{T-1}^d B^\top + \mathbf{E}_n^d.
\end{equation*}
$\widehat{B}$ and $\widehat{\Theta}$ can be obtained by solving an analogously formulated optimization, that is 
\begin{align}\label{opt:VARd}
\begin{split}
(\widehat{B},\widehat{\Theta}):=\argmin\limits_{B\in\R^{p\times dp},\Theta\in \R^{T_d\times K}} \Big\{ \tfrac{1}{2T_d} \smalliii{ \mathbf{X}_T^d - \mathbf{X}^d_{T-1} B^\top - \Theta}_{\F}^2 + \lambda_B\|B\|_1  + \lambda_\Theta\smalliii{\Theta/\sqrt{T_d}}_* \Big\},\\
\text{subject to}~~~~\Theta\in\mathbb{B}_{n_d}(\phi,\mathbf{X}_{T-1}^d).
\end{split}
\end{align}
Empirically at each iteration, $\Theta$ is updated by SVT with hard-thresholding and each row of $B$ is updated via Lasso regression. 

With deterministic realizations based on which we solve the optimization problem, we can obtain essentially the same error bound, with the conditions imposed on the corresponding augmented quantities. Formally, the error bound is given in the next corollary, with a superscript $\star$ associated with the true value of the parameters, $\bunderline{s}:=\sum_{k=1}^d\|B_k^\star\|_0$ being analogously defined as the cardinality of the thresholded support set $\bunderline{S}^\star_\eta := \{(i,j)\,|\,|B_{ij}|>\eta\}$, and $S_\mathbf{X}$ being the sample covariance matrix corresponding to $\mathbf{X}_{T-1}^d$. 
\begin{corollary}[Error bound under $\VAR(d)$ dependence]\label{cor:VARd}
	Suppose the observations stacked in $\mathbf{X}^d_{T-1}$ are deterministic realizations from $\{X_t\}$ process with dynamic given in~\eqref{model:ours-VARd}, and $\mathbf{X}^d_{T-1}$ satisfies the RSC condition with curvature $\alpha_{\text{RSC}}>0$ and a tolerance $\tau_{T_d}$ such that $
	\tau_{T_d}\Big(\bunderline{s} + (2K)\big( \frac{\lambda_{\Theta}}{\lambda_B}\big)^2 \Big) < \min\{\alpha_{\text{RSC}},1\}/64$. 
	Then for any matrix pair $(B^\star,\Theta^\star)$ that drives the dynamic of $X_t$, for estimators $(\widehat{B},\widehat{\Theta})$ obtained by solving~\eqref{opt:VARd} with $\lambda_B$ and $\lambda_\Theta$ chosen such that
	\begin{equation*}
	\lambda_B \geq 2\|(\mathbf{X}^d_{T-1})^\top\mathbf{E}^d_T/T_d\|_\infty+4\phi/\sqrt{T_d(dp)} \quad \text{and} \quad  \lambda_\Theta \geq \Lambda^{1/2}_{\max}(S_{\mathbf{E}}),
	\end{equation*}
	the following error bound holds for some constants $C_1,C_2$ and $C_3$, with $\alpha^\prime:=\min\{\alpha_{\text{RSC}},1\}$:
	{\small
		\begin{equation}
		\smalliii{\Delta_B}_{\F}^2 + \smalliii{\Delta_{\Theta}/\sqrt{T_d}}_{\F}^2 \leq C_1\Big(\frac{\lambda_B}{\alpha^\prime}\Big)^2 \Big\{ \bunderline{s} + \frac{\alpha^\prime}{\lambda_B} \sum_{(i,j)\notin \bunderline{S}^\star_\eta}  |B^\star_{ij}| \Big\} + C_2 \Big(\frac{\lambda_\Theta}{\alpha^\prime}\Big)^2 K + C_3\Big(\frac{\tau_{T_d}}{\alpha^\prime}\Big)(\sum_{(i,j)\notin \bunderline{S}^\star_\eta}  |B^\star_{ij}|)^2.
		\end{equation}}%
\end{corollary}

Next, we verify that the required regularity conditions are satisfied with high probability, and provide high probability bounds for relevant quantities when the data are random realizations from the distribution. Compared with the earlier analyses, a major difference lies in the fact that $X_t$ now has lag-$d$ dependence. However, we note that by considering the stacked transition matrix similar to \citet{basu2015estimation}, each row of $\mathbf{X}_{T-d}^d$ can be viewed as a realization from a $dp$-dimensional process $\bunderline{X}_t$, whose dynamic resembles the previous considered model in Section~\ref{sec:problem}. Specifically, by letting 
\begin{align*}
\bunderline{X}_t &:= \begin{bmatrix}X_t \\ X_{t-1} \\ \vdots  \\ X_{t-d+1}\end{bmatrix}\in\R^{dp}, \qquad 
\quad \bunderline{F}_t : = \begin{bmatrix} F_t \\ 0 \\ \vdots \\ 0 \end{bmatrix}\in\R^{dK}, \qquad 
\bunderline{\epsilon}_t:= \begin{bmatrix} \epsilon_t \\ 0 \\ \vdots \\ 0\end{bmatrix}\in\R^{dp} \\
\end{align*}
and
\begin{align*}
\bunderline{\Lambda} & : = \begin{bmatrix} \Lambda & \mathrm{O} & \dots & \mathrm{O} & \mathrm{O} \\  \mathrm{O} \\ \vdots &  & \mathrm{O}  \\ \mathrm{O} \\ \mathrm{O}
\end{bmatrix}\in\R^{dp\times dK}, \qquad \bunderline{B}:=\begin{bmatrix}
B_1 & B_2 & \cdots & B_{d-1} & B_d \\ \I_p & \mathrm{O} & \cdots & \mathrm{O} & \mathrm{O} \\ \mathrm{O} & \I_p & \cdots & \mathrm{O}  & \mathrm{O} \\ \vdots & \vdots & \ddots & \vdots & \vdots \\ \mathrm{O}& \mathrm{O} & \cdots & \I_p & \mathrm{O} 
\end{bmatrix}\in\R^{dp\times dp},
\end{align*} 
an alternative representation for~\eqref{model:ours-VARd} is given by 
\begin{equation}\label{model:VARd1}
\bunderline{X}_t = \bunderline{\Lambda}\bunderline{F}_t + \bunderline{B}\bunderline{X}_{t-1} + \bunderline{\epsilon}_t.
\end{equation}
Thus, it suffices to verify the RSC condition in an identical way to that in Lemma~\ref{lemma:RSC}, however with the underlying process substituted by $\bunderline{X}_t$; for other quantities such as deviation or the extreme of the eigen-spectrum, the high probability bound should be given based upon $\bunderline{X}_t$ as well. 

\begin{lemma}\label{lemma:RSCVARd} For $\mathbf{X}^d\in\R^{T_d\times dp}$ whose rows are random realizations $\{x_0,\dots,x_{T-1}\}$ of the stable $\{\bunderline{X}_t\}$ process with dynamic given in~\eqref{model:VARd1}. Then there exist positive constants $c_i~(i=0,1,2,3)$ such that with probability at least $1-c_1\exp(-c_2T)$, the RSC condition holds for $\mathbf{X}^d$ with curvature $\alpha_{\text{RSC}}$ and tolerance~$\tau_{T_d}$ satisfying
	\begin{equation*}
	\alpha_{\text{RSC}} = \pi\mathfrak{m}(g_{\bunderline{X}}), \qquad \tau_{T_d} = \alpha_{\text{RSC}}\gamma^2 \frac{\log dp}{2T_d},
	\end{equation*}
	where $\gamma:=54\mathcal{M}(g_{\bunderline{X}})/\mathfrak{m}(g_{\bunderline{X}})$, provided that $n_d\succsim \log (dp)$. 
\end{lemma}

\begin{lemma}\label{lemma:boundXEVARd} There exist positive constants $c_i~(i=0,1,2)$ such that for sample size $T_d\succsim \log (dp)$, with probability at least $1-c_1\exp(-c_2\log (dp))$, the following bound holds:
	\begin{equation*}
	\|(\mathbf{X}^d_{T-1})^\top\mathbf{E}^d_T/T_d\|_\infty \leq c_0\Big(\mathcal{M}(g_{\bunderline{X}}) + \mathcal{M}(g_{\epsilon}) + \mathcal{M}(g_{\bunderline{X},\epsilon}) \Big)\sqrt{\frac{\log (dp)}{T}}.
	\end{equation*}
\end{lemma}

Note that with the definition of $\bunderline{\epsilon}_t$, $\mathcal{M}(g_{\bunderline{X},\bunderline{\epsilon}}) = \mathcal{M}(g_{\bunderline{X},\epsilon})$. Moreover, let $\bunderline{v}_t := \bunderline{\Lambda}\bunderline{F}_t + \bunderline{\epsilon}_t$, then $\bunderline{X}_t = \bunderline{B} \bunderline{X}_t + \bunderline{v}_t $. The following bounds hold for $\mathcal{M}(g_{\bunderline{X}})$ and $\mathfrak{m}(g_{\bunderline{X}})$ \citep{basu2015estimation}:
\begin{equation*}
\mathcal{M}(g_{\bunderline{X}}) \leq \frac{1}{2\pi}\frac{\Lambda_{\max}(\Sigma_{\bunderline{v}})}{\mu_{\min}(\bunderline{\mathcal{B}})}, \qquad \text{and} \qquad \mathfrak{m}(g_{\bunderline{X}})\geq \frac{1}{2\pi}\frac{\Lambda_{\min}(\Sigma_{\bunderline{v}}) }{\mu_{\max}(\mathcal{B}_d)},
\end{equation*}
where we define $\mu_{\max}(\mathcal{B}_d):= \max_{|z|=1} \Lambda_{\max}\Big( \big(\mathcal{B}_d(z)\big)^*\big(\mathcal{B}_d(z)\big) \Big)$, $\mu_{\min}(\bunderline{\mathcal{B}}):=\min\limits_{|z|=1} \Lambda_{\min}\Big( \bunderline{\mathcal{B}}(z)^*\bunderline{\mathcal{B}}(z)  \Big)$, with $\bunderline{\mathcal{B}}(z) = \I_{dp}-\bunderline{B}z$.
\end{document}